\documentclass[%
twocolumn,
superscriptaddress,
amsmath,amssymb,
aps,
]{revtex4-2}

\usepackage{graphicx}
\usepackage{dcolumn}
\usepackage{bm}
\usepackage{qcircuit}
\usepackage{xcolor}
\usepackage{braket}
\usepackage[counterclockwise, figuresleft]{rotating}

\usepackage{lineno}
\usepackage{amsthm}

\usepackage{thmtools}
\usepackage{thm-restate}

\usepackage{tikz}
\usetikzlibrary{shapes.multipart}
\usetikzlibrary{calc,intersections}
\usetikzlibrary{positioning}
\usetikzlibrary{arrows.meta}

\newtheorem{theorem}{Theorem}
\newtheorem{conjecture}{Conjecture}

\newtheorem{lemma}{Lemma}

\theoremstyle{remark}
\newtheorem*{remark}{Remark}

\newcommand{\complexi}{\mathrm{i}}

\newcommand\numberthis{\addtocounter{equation}{1}\tag{\theequation}}

\newcommand{\w}{\mathbf{w}}
\newcommand{\x}{\mathbf{x}}
\newcommand{\y}{\mathbf{y}}
\newcommand{\z}{\mathbf{z}}

\DeclareMathOperator{\Tr}{Tr}
\DeclareMathOperator{\wt}{wt}

\usepackage{hyperref}

\begin{document}

\title{The hardness of quantum spin dynamics}

\author{Chae-Yeun Park}
\email{chae-yeun@xanadu.ai}
\affiliation{Xanadu, Toronto, ON, M5G 2C8, Canada}

\author{Pablo A M Casares}
\affiliation{Xanadu, Toronto, ON, M5G 2C8, Canada}

\author{Juan Miguel Arrazola}
\affiliation{Xanadu, Toronto, ON, M5G 2C8, Canada}

\author{Joonsuk Huh}
\email{joonsukhuh@skku.edu}
\affiliation{Xanadu, Toronto, ON, M5G 2C8, Canada}
\affiliation{Department of Chemistry, Sungkyunkwan University, Suwon 16419, Korea}
\affiliation{SKKU Advanced Institute of Nanotechnology (SAINT), Sungkyunkwan University, Suwon 16419, Korea}
\affiliation{Institute of Quantum Biophysics, Sungkyunkwan University, Suwon 16419, Korea}

\date{\today}

\begin{abstract}
Recent experiments demonstrated quantum computational advantage in random circuit sampling and Gaussian boson sampling.
However, it is unclear whether these experiments can lead to practical applications even after considerable research effort.
On the other hand, simulating the quantum coherent dynamics of interacting spins has been considered as a potential first useful application of quantum computers, providing a possible quantum advantage.
Despite evidence that simulating the dynamics of hundreds of interacting spins is challenging for classical computers, concrete proof is yet to emerge. 
We address this problem by proving that sampling from the output distribution generated by a wide class of quantum spin Hamiltonians is a hard problem for classical computers.
Our proof is based on the Taylor series of the output probability, which contains the permanent of a matrix as a coefficient when bipartite spin interactions are considered.
We devise a classical algorithm that extracts the coefficient using an oracle estimating the output probability.
Since calculating the permanent is \#P-hard, such an oracle does not exist unless the polynomial hierarchy collapses.
With an anticoncentration conjecture, the hardness of the sampling task is also proven.
Based on our proof, we estimate that an instance involving about 200 spins will be challenging for classical devices but feasible for intermediate-scale quantum computers with fault-tolerant qubits.  
\end{abstract}

\maketitle

\section{Introduction}

There is great interest in demonstrating that quantum computers can outperform powerful supercomputers. The most successful paradigm so far involves showing that sampling from the output distribution of certain quantum circuits is classically difficult~\cite{hangleiter2023computational}.
Three prominent examples of quantum sampling problems are \textsf{BosonSampling}~\cite{Aaronson2011}, random circuit sampling (\textsf{RCS})~\cite{boixo2018characterizing}, and instantaneous quantum polynomial-time circuits (\textsf{IQP})~\cite{bremner2016average}.
Theoretical works assert that a classical computer can only efficiently simulate the same tasks if the polynomial hierarchy (\textsf{PH}) collapses to the third level.
Hence, under the widely believed conjecture that \textsf{PH} is infinite, those sampling tasks are classically intractable.

Based on these theoretical suggestions, several experiments have realized \textsf{RCS} and a Gaussian variant of \textsf{BosonSampling} using noisy qubits~\cite{arute2019quantum} or photonic modes~\cite{zhong2020quantum,madsen2022quantum}.
However, the main drawback of these quantum advantage results is that they have little to do with practical applications.
These circuits are designed to generate nearly random bitstrings, and despite several efforts~\cite{bromley2020applications}, we have yet to find convincing practical applications.

On the other hand, solving quantum many-body problems is expected to be the first useful application of quantum computers~\cite{bravyi2022future,beverland2022assessing}.
Quantum many-body Hamiltonians describing interacting quantum particles are essential in physics, serving as models for various ranges of systems, from nuclei to materials~\cite{fetter1971quantum}.
As we are often interested in low-temperature physics, finding the properties of the lowest energy state, i.e., the ground state, of quantum many-body Hamiltonians has been the central problem over decades~\cite{kohn1965self}.
Indeed, many important questions in chemistry and materials science, such as the electronic, magnetic, and mechanical properties of molecules and materials, are answered by attributes of the ground state of these systems~\cite{ashcroft1976solid,levine2014quantum}.

Due to its importance, significant effort has been devoted to computing properties of the ground state, such as the ground-state energy~\cite{avella2012strongly}.
However, it is well agreed that solving the ground-state energy problem of a large quantum many-body Hamiltonian with hundreds of interacting particles is highly challenging to a classical computer~\cite{elfving2020will}. 
This fact is further supported by computational complexity theory, where it has been shown that finding the ground-state energy is \textsf{QMA-hard}~\cite{kitaev2002classical} --- a quantum version of \textsf{NP-hard}. 
Similar to \textsf{NP-hard} problems that are regarded as difficult for classical computers, it is widely believed that \textsf{QMA-hard} problems are challenging even for quantum computers~\cite{Bookatz14QMA}.
Hence, we do not expect an efficient classical algorithm for finding the ground-state energy of an arbitrary Hamiltonian to exist.

Recently, the quantum coherent dynamics of many-body systems have also gained considerable attention as they are related to experiments with engineered quantum systems~\cite{gross2017quantum} and also provide an avenue to understand quantum chaos~\cite{srednicki1994chaos}.
The dynamical problem can, in principle, be more complicated than the ground-state problem because it involves a manifold containing many eigenstates, including the ground state.
Despite such evidence, the hardness of quantum dynamics has only been proven for a limited class of Hamiltonians, such as Ising variants with commuting terms~\cite{gao2017quantum,bermejo2018architectures,haferkamp2020closing,Fefferman2017}, or a spin model that can be translated into interacting bosons~\cite{Peropadre2017}. 
The hardness proof of quantum dynamics for a broader class of Hamiltonians would be significant given that simulating the dynamics of quantum systems is a potential first practical application of quantum computing hardware, both for noisy~\cite{clinton2022towards} and early fault-tolerant quantum computers~\cite{beverland2022assessing}.
A hardness proof can also guide the type of experiment that should be tackled to claim quantum computational advantage.
A recent experiment~\cite{kim2023evidence} and follow-up classical simulation results~\cite{tindall2023efficient,kechedzhi2023effective,beguvsic2023fast} demonstrate that lots of subtleties exist in proving quantum computational advantage and why a concrete theoretical guide is vital.

In this work, we prove the classical hardness of simulating quantum spin dynamics. Our target Hamiltonians include the Ising, XX, and Heisenberg models, some of the most widely studied many-particle spin-$1/2$ Hamiltonians. 
While some variants of \textsf{IQP} circuits already imply that computing the output distribution at a certain time is classically hard for some Ising-type Hamiltonians~\cite{gao2017quantum,Fefferman2017,bermejo2018architectures,haferkamp2020closing}, these results are still far from physically interesting scenarios as the complexity is only proved at a single time point~\cite{gao2017quantum,bermejo2018architectures,haferkamp2020closing} or for an exponentially short time~\cite{Fefferman2017}.
The dynamics of those Hamiltonians are also rather simple due to their commuting structure.

On the other hand, Peropadre et al.~\cite{Peropadre2017} studied the classical hardness of sampling from the output distributions of the XY model, whose terms are non-commuting. 
Nevertheless, their hardness result is only valid in the limit where spins approximate hard-core bosons.
This implies that their protocol would require hundreds of thousands of spins to claim the quantum computational advantage.
In contrast, we prove that approximating the output distribution is classically hard for a wide class of non-commuting spin Hamiltonians without taking such a limit.
We also argue that our protocol only requires about 200 spins to demonstrate the advantage, which is more experimentally feasible.

We utilize recently developed techniques for bounding errors from truncated Taylor series~\cite{bouland2019complexity,movassagh2023hardness,bouland2022noise,krovi2022average} to prove hardness.
Precisely, we show that an efficient classical algorithm estimating a polynomial function within a small additive error can also estimate any of its high-order derivatives. 
This result is particularly interesting for a certain class of Hamiltonians where the high-order derivative of the output distribution is given by the permanent of a matrix.
Under the well-established hardness result that computing permanent of random Gaussian matrices is \textsf{\#P-hard}~\cite{Aaronson2011}, we prove that estimating the output distribution for these Hamiltonians is also classically intractable.

With an anticoncentration conjecture, we also prove that an efficient classical algorithm that can sample from the output distribution does not exist unless the \textsf{PH} collapses.
The main technical challenge here is that, unlike other known sampling hardness results, our problem does not have the hiding property.
Roughly speaking, the hiding property refers to the fact that the output bitstring can be hidden by randomizing the circuit instances, i.e., a property of the output distribution averaged over random circuit instances does not depend on the output bitstrings~\cite{Aaronson2011,bremner2016average,boixo2018characterizing}. 
We accommodate this problem by proving an average-case hardness result that works even when the output distribution anticoncentrates only for a small fraction of outcomes, at the expense of a smaller $L_1$ distance between the true distribution and the distribution we sample from.
We provide numerical evidence that our models can show such a relaxed anticoncentration property, completing arguments for sampling hardness.

We further discuss an experimental implementation of our protocol based on fault-tolerant intermediate-scale quantum (ISQ) hardware~\cite{ISQblog}. 
By implementing the time-evolution operator using the Trotter decomposition, we roughly estimate that a circuit with around two hundred logical qubits and less than a billion gates is sufficient to demonstrate quantum advantage using our protocol, which could be feasible for the first-generation of fault-tolerant quantum computers~\cite{bourassa2021blueprint,beverland2022assessing,bravyi2022future,campbell2021early,daley2022practical}.

\section{Preliminaries}
\begin{figure*}[htb]
    \centering
    \includegraphics[width=0.9\linewidth]{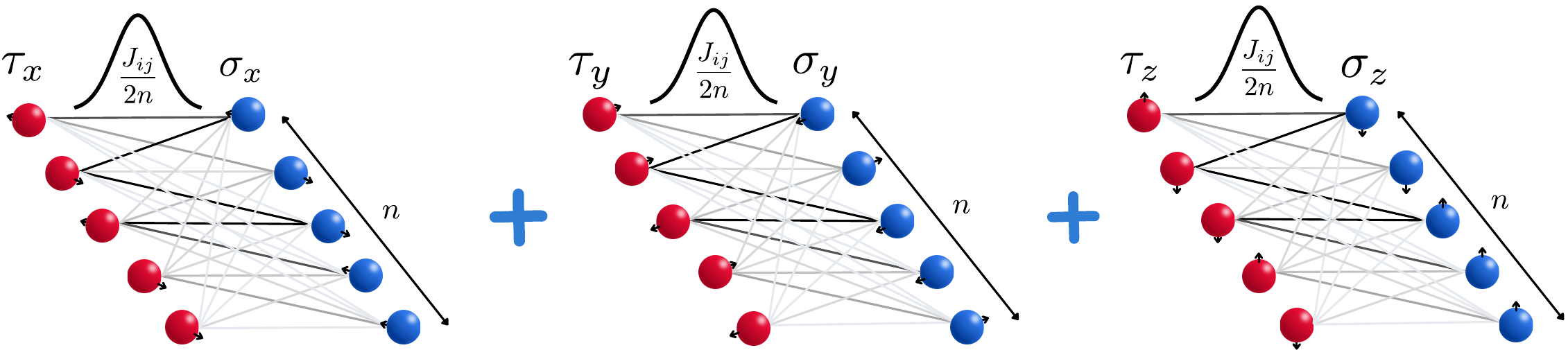}
    \caption{ A pictorial description of the Heisenberg model Hamiltonian ($H_{4}$) in Eq.~\eqref{eq:Hamiltonian4} with bipartite networks between spin groups $\sigma$ and $\tau$ composed of $n$ spins. $J_{ij}$ denotes a random coupling strength between $\sigma^{(i)}$ and $\tau^{(j)}$ drawn from a normal distribution. Coupling strengths are denoted by the thickness of the corresponding edges.
    }
    \label{fig:Heisenberg_model}
\end{figure*}
Throughout the paper, we consider spin-$1/2$ models defined over a bipartite graph, i.e., a graph with two sets of vertices $U$ and $V$ such that there is no edge between the vertices within each set. We further restrict $U$ and $V$ to have the same number of vertices, i.e., $|U|=|V|=n$. Thus, our system comprises a total of $2n$ spins.
We denote the spin variables, i.e., Pauli operators, for the $i$-th ($j$-th) spin of $U$ ($V$) by $\sigma^{(i)}_{x,y,z}$ ($\tau^{(j)}_{x,y,z}$), where the subscript $x,y,z$ refers to the type of Pauli operator. We consider four types of Hamiltonians defined as follows:
\begin{align}
    H_{1} &= \sum_{i,j=1}^{n} \frac{J_{ij}}{n} \sigma^{(i)}_x \tau^{(j)}_x  ,\label{eq:Hamiltonian1}\\
    H_{2} &= \sum_{i,j=1}^{n} \frac{J_{ij}}{n} \bigl[ \sigma^{(i)}_x \tau^{(j)}_x + \sigma^{(i)}_z \tau^{(j)}_z \bigr] ,\label{eq:Hamiltonian2}\\
    H_{3} &= \sum_{i,j=1}^{n} \frac{J_{ij}}{2n} \bigl[ \sigma^{(i)}_x \tau^{(j)}_x + \sigma^{(i)}_y \tau^{(j)}_y \bigr] ,\label{eq:Hamiltonian3}\\
    H_{4} &= \sum_{i,j=1}^{n} \frac{J_{ij}}{2n} \bigl[ \sigma^{(i)}_x \tau^{(j)}_x + \sigma^{(i)}_y \tau^{(j)}_y + \sigma^{(i)}_z \tau^{(j)}_z \bigr].
    \label{eq:Hamiltonian4}
\end{align}
Hamiltonians $H_1$, $H_3$, and $H_4$ are often called the Ising, the XX, and the Heisenberg models, respectively. See Fig.~\ref{fig:Heisenberg_model} for an illustration of the Heisenberg model Hamiltonian.
We will show that computing the outcome probability distributions from the unitary evolution generated by these Hamiltonians is classically hard.

We additionally define the following notations for binary vectors for convenience:
\begin{itemize}
    \item We use bold letters to denote binary vectors in $\mathbb{Z}_2^{2n}$. For example, $\mathbf{x} = \{x_1, \cdots, x_{2n}\}$ is a vector and $\ket{\mathbf{x}}=\ket{x_1}\cdots \ket{x_{2n}}$ is a product state in the computational basis.
    \item We denote the first and last $n$ bits of $\mathbf{x} \in \mathbb{Z}_2^{2n}$ by $\mathbf{x}^{\sigma} := \{x_1, \cdots x_n\}$ and $\mathbf{x}^{\tau} := \{x_{n+1}, \cdots x_{2n}\}$, respectively.
    \item The Hamming weight of a vector $\wt(\mathbf{x})$ is the number of 1s in the vector.
    \item We denote a set of binary vectors with $\wt(\mathbf{x}^{\sigma})=m$ and $\wt(\mathbf{x}^{\tau}) = n-m$ by $X_{m}$, i.e., $X_{m}=\{\mathbf{x} \in \mathbb{Z}_2^{2n}: \wt(\mathbf{x}^{\sigma})=m \text{ and } \wt(\mathbf{x}^{\tau})=n-m \}$.
    \item $X:=\{\mathbf{x}\in \mathbb{Z}_2^{2n}: \wt(\mathbf{x}) = n\}$ is a set of binary vectors with Hamming weight $n$. Note that $X=\cup_{i=0}^n X_i$.
\end{itemize}
We also utilize big $O$ and related notations, as summarized in Appendix~\ref{app:big_o_notation}.

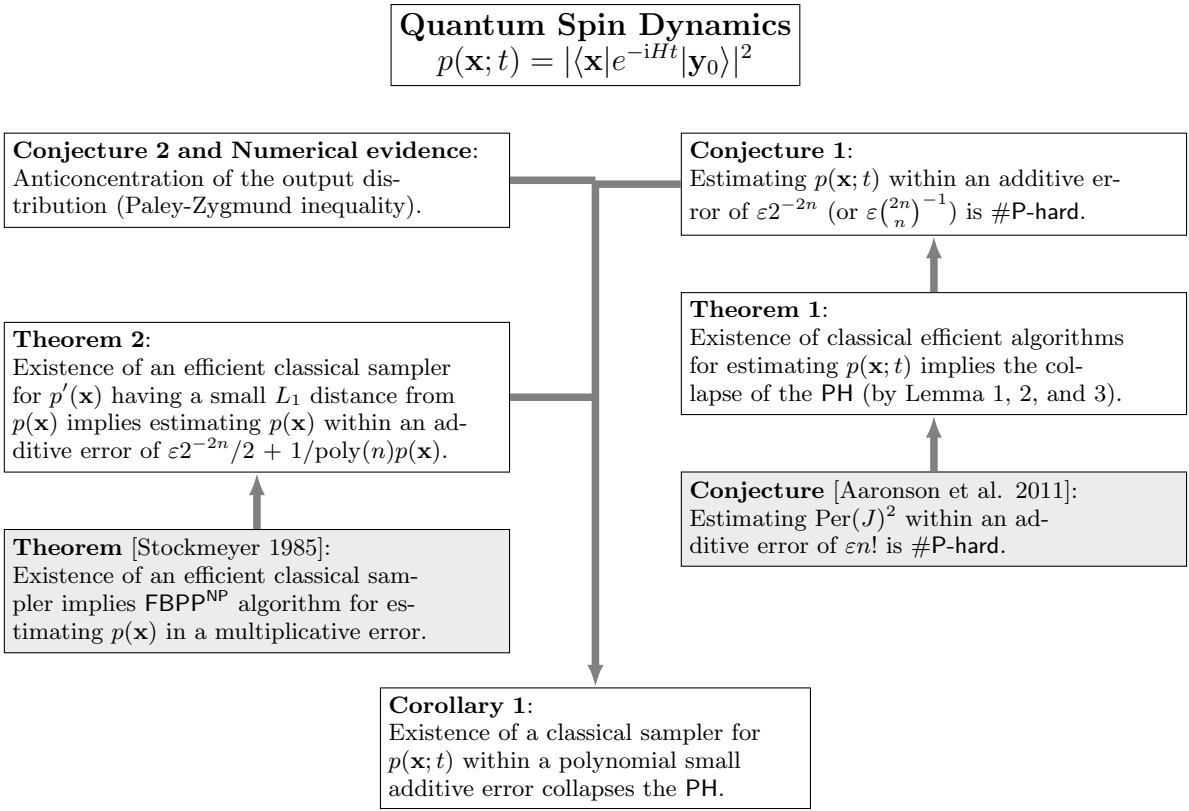
\begin{figure*}[htb]
\centering
\begin{tikzpicture}
	\node[draw] (1) at (10.0, 20.0) [align=center, text centered] {\large \textbf{Quantum Spin Dynamics} \\ \large $p(\mathbf{x};t) = |\langle \mathbf{x} | e^{-\mathrm{i}Ht}| \mathbf{y}_0 \rangle|^2$};

	\node[draw] (conj1) at (14.5, 18.5) [text width=6.5cm, anchor=base] {\textbf{Conjecture 1}: \\ Estimating $p(\mathbf{x};t)$ within an additive error of $\varepsilon 2^{-2n}$ (or $\varepsilon {2n \choose n}^{-1}$) is \textsf{\#P-hard}.};
 
	\node[draw] [below=0.75 of conj1] (thm1) [text width=6.5cm] {\textbf{Theorem 1}: \\ Existence of classical efficient algorithms for estimating $p(\mathbf{x};t)$ implies the collapse of the \textsf{PH} (by Lemma 1, 2, and 3).};
 
	\node[draw] [below=0.75cm of thm1] (4) [fill=gray!15,text width=6.5cm] {\textbf{Conjecture} [Aaronson et al.  2011]: \\ Estimating $\mathrm{Per}(J)^2$ within an additive error of $\varepsilon n!$ is \textsf{\#P-hard}.};

	\node[draw] (conj3) at (5.5, 18.5) [text width=6.5cm, anchor=base] {\textbf{Conjecture 2 and Numerical evidence}: \\ Anticoncentration of the output distribution (Paley-Zygmund inequality).};

	\node[draw] [below=1.25cm of conj3] (thm3) [text width=6.5cm] {\textbf{Theorem 2}: \\ Existence of an efficient classical sampler for $p'(\mathbf{x})$ having a small $L_1$ distance from $p(\mathbf{x})$ implies estimating $p(\mathbf{x})$ within an additive error of $\varepsilon 2^{-2n}/2 + 1/\mathrm{poly}(n) p(\mathbf{x})$.};
	\node[draw] [below=0.75cm of thm3] (stockmeyer) [fill=gray!15,text width=6.5cm] {\textbf{Theorem} [Stockmeyer 1985]: \\ Existence of an efficient classical sampler implies $\mathsf{FBPP}^\mathsf{NP}$ algorithm for estimating $p(\mathbf{x})$ in a multiplicative error.};

	\node[draw] [below=8cm of 1] (cor1) [text width=5.5cm] {\textbf{Corollary 1}: \\ Existence of a classical sampler for $p(\mathbf{x};t)$ within a polynomial small additive error collapses the \textsf{PH}.};

	\draw[gray, line width=1mm, arrows = {-Latex[length=0pt 3 0, angle'=40]}] let \p1 = (thm1.north), \p2 = (conj1.south), \n1 = {(\y1+\y2)/2} in (thm1.north) -- ($(\x1, \n1)$) --  ($(\x2, \n1)$) -- (conj1.south);
	\draw[gray, line width=1mm, arrows = {-Latex[length=0pt 3 0, angle'=40]}] (4.north) -- (thm1.south);
	
	\draw[gray, line width=1mm, arrows = {-Latex[length=0pt 3 0, angle'=40]}] (stockmeyer.north) -- (thm3.south);

	\draw[gray, line width=1mm] (conj3.east) -- (1 |- conj3);
	\draw[gray, line width=1mm] (thm3.east) -- (1 |- thm3);


	\draw[gray, line width=1mm, arrows = {-Latex[length=0pt 3 0, angle'=40]}] (conj1.west) -- (1 |- conj1) -- (cor1);

\end{tikzpicture}
\caption{Diagrammatic overview of relations between conjectures and theorems. Boxes filled in gray indicate statements not developed in this study.}
\label{fig:QSDdiagram}
\end{figure*}

\section{Theoretical results and techniques}
In this section, we introduce the classical hardness of approximately estimating the output probability of quantum spin dynamics generated by four types of Hamiltonians, $H_{1,2,3,4}$, (Theorem~\ref{thm:classical_hardness}) and of sampling from the corresponding output probability distribution (Corollary~\ref{col:complexity_small_tvd_sampling}).

The central observation for our hardness proof is that the leading coefficient of the Taylor expansion of the output probability contains the permanent of a matrix.
By utilizing recently developed \textit{robust polynomial regression} techniques~\cite{bouland2019complexity,movassagh2023hardness,bouland2022noise,krovi2022average}, we show that the efficient average-case estimation of the output probabilities implies the estimation of the permanent, which is \textsf{\#P-hard}.
In addition, with the aid of an anticoncentration conjecture, we argue that the corresponding sampling task is also classically challenging. The main conjectures and theorems (see Fig.~\ref{fig:QSDdiagram}) are followed by evidence and proofs using the listed mathematical tools.

\subsection{Classical hardness of computing the output distribution}

We introduce a key conjecture below, which states that estimating the output probability of quantum spin dynamics within an exponentially small additive error is \textsf{\#P-hard}. 
Many previous studies have conjectured that approximating the output probability is \textsf{\#P-hard} on average for classes of quantum circuits~\cite{Aaronson2011,bremner2016average,boixo2018characterizing,bermejo2018architectures}. 
Our conjecture extends these arguments to quantum spin dynamics.
This conjecture is supported by one of our main results, Theorem~\ref{thm:classical_hardness}, presented later in this section.
\begin{conjecture} \label{conj:hardness_output_probs}
    Let $p_{1,2,3,4}(\mathbf{x};J;t):=|\braket{\mathbf{x}|e^{-\complexi H_{1,2,3,4}t}|\mathbf{y}_0}|^2$ where $\ket{\mathbf{y}_0} = \ket{\mathbf{0}}^\sigma \ket{\mathbf{1}}^\tau$, i.e., we can write $\mathbf{y}_0 = \mathbf{0}^{\sigma}\mathbf{1}^{\tau}=0\cdots01\cdots1$.
    For $\mathbf{x} \in X_m$ with $m  \geq \sqrt{n}$ and $t_0=\Theta(\log n)$, the following statements hold.  
    \begin{itemize}
        \item \emph{ \textsf{Problem I}:} It is \emph{\textsf{\#P-hard}} to approximate $p_1(\mathbf{x};J;t_0)$ or $p_2(\mathbf{x};J;t_0)$ within an additive error of  $\varepsilon 2^{-2n}$ with any constant probability of $p > 0$ over $J\sim \mathcal{N}(0,1)^{n \times n}$.
        \item \emph{\textsf{Problem II}:} It is \emph{\textsf{\#P-hard}} to approximate $p_3(\mathbf{x};J;t_0)$ or $p_4(\mathbf{x};J;t_0)$ within an additive error of $\varepsilon {2n \choose n}^{-1}$ with any constant probability of $p > 0$ over $J\sim \mathcal{N}(0,1)^{n \times n}$.
    \end{itemize}
\end{conjecture}

We provide the following three remarks in order.
First, the Hamming weight $m$ must be large for $\mathbf{x} \in X_m$.
An intuitive way to understand this condition is that a large value of $m$ ensures that $p(\mathbf{x};J;t)$ cannot be efficiently computed using a low-degree expansion at $t=0$.
This is because the overlap, $\braket{\mathbf{x} | H^k| \mathbf{y}_0}$, is non-zero only when $k \geq m$.
Thus, if one wants to compute $p(\mathbf{x};J;t)$ using the polynomial expansion in $t$, e.g., using the Krylov subspace method~\cite{liesen2013krylov}, $p(\mathbf{x};J;t)$ must be expanded to an order that is sufficiently larger than $m$; this is classically difficult.

Second, the difference between the required errors $\varepsilon 2^{-2n}$ and $\varepsilon {2n \choose n}^{-1}$ originates from the fact that $H_{3}$ and $H_{4}$ have a global $U(1)$ symmetry, but $H_{1}$ and $H_{2}$ do not.
Given that there exist $2n$ spins, and the initial state resides in the eigenspace of $J_z:=\sum_i \sigma^{(i)}_{z}+\sum_j \tau^{(j)}_{z}$ with a corresponding eigenvalue of zero, we know that the final state is also within this subspace since $[J_z, H_{3,4}] = 0$. 
As the dimension of the subspace is $D_{II}:=|X|={2n \choose n}$, we have $\mathbb{E}_{\mathbf{x} \in X}[p(\mathbf{x};J;t)] = {2n \choose n}^{-1}$.
In contrast, the dimension $D_I$ of the allowed space is $2^{2n}$ for $H_{1}$ and $H_{2}$, which do not have $U(1)$ symmetry.
In short, the conjecture states that approximating $p(\mathbf{x};J;t)$ within a small factor $\varepsilon$, on an average over $\mathbf{x}$, is classically hard.

Lastly, we require $t=\Theta(\log n)$ as the output distribution anticoncentrates for such $t$ (throughout the paper, `$\log$' stands for the natural logarithm, unless otherwise noted).
The anticoncentration property is characterized by the fact that the output distribution $p(\mathbf{x};J;t)$ only fluctuates near (up to a constant factor) its averaged value, i.e., $2^{-2n}$ and ${2n \choose n}^{-1}$ for $H_{1,2}$ and $H_{3,4}$, respectively.
While the anticoncentration by itself is not a sufficient condition for classical hardness, it supports the average-case hardness in the sense that a trivial estimation of $p(\mathbf{x};J;t) = 0$ must be incorrect~\cite{hangleiter2023computational}. 
Numerical evidence for anticoncentration is provided in Sec.~\ref{sec:numerics_anticon}.

We now present our main theorem supporting Conjecture~\ref{conj:hardness_output_probs}. 
Most importantly, we show the classical hardness results given that the estimation of permanents of Gaussian matrices is \textsf{\#P-hard} as in \textsf{BosonSampling}~\cite{Aaronson2011}.

\begin{restatable}{theorem}{ThmClassicalHardness} \label{thm:classical_hardness}
    Suppose that $J\sim \mathcal{N}(0,1)^{n \times n}$ is a randomly sampled Gaussian matrix, $H$ is one of $H_{1,2,3,4}$, and $p(\mathbf{x};J;t):=|\braket{\mathbf{x}| e^{-\complexi Ht}|\mathbf{y}_0}|^2$.
    Then for a given $\mathbf{x} \in X_m$ with $m \geq \sqrt{n}$, $t_0 = O(\log n)$, and a constant $0 < \Delta < 1$, there exists a positive constant $c_1$ such that the following statement holds: If there exists an efficient classical algorithm that estimates $p(\mathbf{x};J;t_0)$ within an additive error of $n^{-c_1 n^3 \log n }$ with a probability of $\eta \geq 3/4$ over $J$, then $\mathsf{FBPP}=\mathsf{\#P}$.
\end{restatable}
Theorem~\ref{thm:classical_hardness} proves a weaker version of Conjecture~\ref{conj:hardness_output_probs}.
An oracle formulation of Conjecture~\ref{conj:hardness_output_probs} shows the difference between the two.
For example, \textsf{Problem I} of Conjecture~\ref{conj:hardness_output_probs} is equivalent to the following statement: ``If there exists an efficient classical algorithm that estimates $p(\mathbf{x};J;t_0)$ within an additive error of $\varepsilon 2^{-2n}$, \textsf{FP=\#P}.''
Compared to Conjecture~\ref{conj:hardness_output_probs}, the theorem demands a stronger oracle for $p(\mathbf{x};J;t_0)$, in the sense that it requires a smaller error $n^{-c_1 n^3 \log n}$ and a larger probability ($\geq 3/4$) over $J$. 
It also requires a slightly stronger relation between the complexity classes for the hardness that $\mathsf{FBPP} \neq \mathsf{\#P}$, which is satisfied when the \textsf{PH} does not collapse to the second level~\cite{arora2009computational} (still widely believed to be true).

For the proof of Theorem~\ref{thm:classical_hardness}, we utilize a polynomial expansion of $p(\mathbf{x};t) = |\braket{\mathbf{x}| e^{-\complexi Ht}|\mathbf{y}_0}|^2$ for $\mathbf{x} \in X_m$ (where we exclude $J$ from the arguments when they are obvious). The expansion contains the permanent of $\tilde{J}$, a $m \times m$ submatrix of $J$.
The permanent of a matrix $A=(a_{ij})$ is defined as $\mathrm{Per}(A):=\sum_{s \in S_n} \prod_{i=1}^n a_{i, s(i)}$ where $S_n$ is the symmetric group with $n$ elements.
Namely, the polynomial expansion of $p(\mathbf{x};t)$ is given by 
\begin{align}
    p(\mathbf{x};t) = \frac{\mathrm{Per}(\tilde{J})^{2}}{n^{2m}} t^{2m}+a_{2m+1} t^{2m+1} + \cdots .
\end{align}
We provide a detailed derivation of this property in Appendix~\ref{app:complexity_output_dist_proof}.
Given that estimating the permanent of a Gaussian random matrix is \textsf{\#P-hard},
estimating the $2m$-th coefficient of the polynomial expansion (i.e., the coefficient for $t^{2m}$) is a classically hard problem for $m \geq \sqrt{n}$.

We then devise an \textsf{FBPP} algorithm that estimates the coefficient for $t^{2m}$ when an oracle estimating $p(x;J;t)$ near $t=t_0$ is given.
Our algorithm is based on recently developed \textit{robust polynomial regression} techniques~\cite{sherstov2012making,movassagh2023hardness,kane2017robust,bouland2019complexity,bouland2022noise,krovi2022average,kondo2022quantum}.
We combine our algorithm with the fact that an oracle that estimates $p(\mathbf{x};J;t_0)$ with probability $\eta$ over $J \sim \mathcal{N}(0,1)^{n \times n}$ can also estimate $p(\mathbf{x};J;t)$ with probability $\eta - (3/2)n\sqrt{|(t/t_0)^2 - 1|}$ to complete the proof (see Appendix~\ref{app:complexity_output_dist_proof} for details).

While the proof is constructed around Gaussian random matrices, our model is flexible enough to accommodate encoding any matrix through the couplings $J$. This allows for more versatility in proof techniques, that could be lifted in future work to further strengthen hardness proofs.

\subsection{Classical hardness of sampling from the distribution} \label{sec:hardness_sampling}
Now, let us discuss the corresponding sampling problem.
Following usual arguments for sampling hardness in Refs.~\cite{Aaronson2011,bremner2016average,boixo2018characterizing}, we combine Stockmeyer's theorem~\cite{Stockmeyer1985} and an anticoncentration conjecture (see also Appendix~\ref{app:hardness_sampling_proof}).
In these arguments, Stockmeyer's theorem provides an $\mathsf{FBPP}^\mathsf{NP}$ algorithm for estimating the output probability within an inverse polynomial multiplicative error, 
which is further reduced to an additive error utilizing the anticoncentration property.
Anticoncentration is characterized by the fact that, for any $\mathbf{x}$, the output probability $p_C(\mathbf{x})$ for a circuit $C$ with $n$ qubits fluctuates only near $2^{-n}$ over random circuit instances $C$.

However, these techniques cannot be directly applied to our case, as our problem does not possess the hiding property, unlike those examples.
The hiding property holds when, for a given circuit $C(\theta)$ with a parameter $\theta$ sampled from a random distribution, the output probability for a given outcome string $\mathbf{x}$ is the same as that for $\mathbf{0}$ but with a different parameter, i.e., $|\braket{\mathbf{x}|C(\theta)|\mathbf{0}}|^2 = \braket{\mathbf{0}|C(\theta')|\mathbf{0}}|^2$.
As a result, properties of the output distribution averaged over the circuit parameter $\theta$ do not depend on the outcome bitstring $\mathbf{x}$.
For example, the average-case hardness of $|\braket{\mathbf{0}|C(\theta)|\mathbf{0}}|^2$ over $\theta$ holds for $|\braket{\mathbf{x}|C(\theta)|\mathbf{0}}|^2$ with any $\mathbf{x}$, and the anticoncentration property of $|\braket{\mathbf{0}|C(\theta)|\mathbf{0}}|^2$ implies that $|\braket{\mathbf{x}|C(\theta)|\mathbf{0}}|^2$ also anticoncentrates for any $\mathbf{x}$.

Compared to those cases, our complexity results (Theorem~\ref{thm:classical_hardness}) are only satisfied for $\mathbf{x} \in X_{m}$ with sufficiently large $m$.
The main result obtained in this section is that such restricted complexity results are sufficient to prove the difficulty of sampling when the error between the distribution we sample from and the true distribution $p(\mathbf{x};J;t)$ is inverse polynomially small, provided that the output distribution anticoncentrates at least for a constant fraction of $\mathbf{x} \in X_{n/2}$.

Specifically, we obtain the following theorem by applying Stockmeyer's theorem~\cite{Stockmeyer1985} to $p(\mathbf{x})$ with $\mathbf{x} \in X_{n/2}$ (see Appendix~\ref{app:hardness_sampling_proof} for details).

\begin{restatable}{theorem}{applystockmeyer}
\label{thm:apply_stockmeyer}
    Suppose that there exists a classically efficient oracle, $\mathcal{O}$, that can sample from a probability distribution $p'(\mathbf{x})$ which approximates $p(\mathbf{x})$ within an additive error of $\nu$ in $L_1$ norm.
    We further assume that $p(\mathbf{x})$ and $p'(\mathbf{x})$ have the same domain, i.e., $p'(\mathbf{x}) = 0$ for all $\mathbf{x} \notin X$ when $H= H_{3,4}$.
    Then for $\gamma \in (0, 1)$, there is an $\mathsf{FBBP}^\mathsf{NP}$ machine that estimates $p(\mathbf{x})$ within an additive error of
    \begin{itemize}
        \item $\epsilon_{I}=(1+g)\gamma^{-1} \nu (2n) 2^{-2n} + g p(\mathbf{x})$ ,
        \item $\epsilon_{II}=(1+g) (\sqrt{\pi} \gamma/2)^{-1} \nu n^{1/2} {2n \choose n}^{-1} + g p (\mathbf{x})$  ,
    \end{itemize}
    for $H_{1,2}$ and $H_{3,4}$, respectively, 
    with a success probability of $\geq 1 - \gamma$ for $\mathbf{x} \in X_{n/2}$, where $g=1/\mathrm{poly}(n)$.
\end{restatable}

We also need the following anticoncentration conjecture for the hardness of sampling.

\begin{restatable}[Anticoncentration]{conjecture}{anticoncentration}
\label{conj:anticoncentration}
Let $p(\mathbf{x};J;t)=|\braket{\mathbf{x} | e^{-\complexi Ht} | \mathbf{y}_0}|^2$ for $\mathbf{y}_0$ given in Conjecture~\ref{conj:hardness_output_probs}.
Then, there exists $t_1 = \Theta(\log n)$ such that for all $t > t_1$, the following holds for a constant fraction, $\zeta$, of $\mathbf{x} \in X_{n/2}$:
\begin{itemize}
    \item There exist universal constants $0 < \alpha, \beta <1$ such that
    \begin{align}
        \underset{J \sim \mathcal{N}(0,1)^{n \times n}}{\mathrm{Pr}} \Bigl[ p(\mathbf{x};J;t) > \alpha 2^{-2n} \Bigr] \geq \beta ,
    \end{align}
    \item There exist universal constants $0 < \alpha, \beta < 1$ such that
    \begin{align}
        \underset{J \sim \mathcal{N}(0,1)^{n \times n}}{\mathrm{Pr}} \Bigl[ p(\mathbf{x};J;t) > \alpha {2n \choose n}^{-1} \Bigr] \geq \beta ,
    \end{align}
\end{itemize} 
for $H_{1,2}$ and $H_{3,4}$, respectively.
\end{restatable}

Here, $t_1 = \Theta(\log n )$ is obtained from the numerical observation presented in the next section. 
As we consider $U=e^{-iHt}$ with $\Vert H \Vert =O(n)$, our condition states that anticoncentration holds when $\Vert H\Vert t = O(n \log n)$.
This is consistent with the condition that the sparse \textsf{IQP} demonstrates anticoncentration~\cite{bremner2017achieving}, the circuit of which can be written as $e^{-i H'}$ with $\Vert H' \Vert = O(n \log n)$.

Mark et al.~\cite{mark2023benchmarking} provides further evidence for Conjecture~\ref{conj:anticoncentration} at least for $H_{2,3,4}$ and large $t$. 
Specifically, the study showed that the output distribution converges to the Porter--Thomas distribution as $t \rightarrow \infty$ when the eigenvalues of the Hamiltonian follow a reasonable condition, which is likely to be satisfied for non-integrable Hamiltonians.
As $H_{2,3,4}$ are non-integrable, their results suggest that the output distribution of $H_{2,3,4}$ anticoncentrates for large $t$ since the Porter--Thomas distribution implies anticoncentration.
Although this result only applies to sufficiently large $t$, given that anticoncentration is weaker than the Porter--Thomas distribution, we expect that the output distribution can anticoncentrate for smaller $t$.
In Appendix~\ref{app:anticon_ising}, we also prove a result supporting anticoncentration for $H_1$.
We particularly show that the time-averaged value of $p(\mathbf{x};J;t)$ must be $2 \times 2^{-2n}$ for all $\mathbf{x} \in X$ regardless of $J$.
Thus, if $t$ is sufficiently large to equilibrate all $p(\mathbf{x};J;t)$, the output distribution anticoncentrates.

By combining Theorem~\ref{thm:apply_stockmeyer} and Conjecture~\ref{conj:anticoncentration}, we obtain the following corollary  (see Appendix~\ref{app:hardness_sampling_proof} for details).

\begin{restatable}{corollary}{ComplexitySmallTVDSampling}
\label{col:complexity_small_tvd_sampling}
    Suppose that there exists $t_0 = \Theta(\log n)$ such that Conjectures~\ref{conj:hardness_output_probs} and \ref{conj:anticoncentration} hold. 
    If there exists a classically efficient algorithm that can sample from a probability distribution which approximates $p_{1,2}(\mathbf{x};J;t_0)$ $(p_{3,4}(\mathbf{x};J;t_0))$ within an additive error of $\varepsilon \zeta n^{-1}/8$ $(\varepsilon \zeta\sqrt{\pi} n^{-1/2}/8)$, then there exists an $\mathsf{FBPP}^{\mathsf{NP}}$ algorithm that can approximate $p(\mathbf{x};J;t_0)$ within an additive error of $\varepsilon 2^{-2n}$ $(\varepsilon {2n \choose n}^{-1})$ with probability greater than $1 - \beta$ over $J$.
    As it is a \emph{\textsf{\#P-hard}} problem according to Conjecture~\ref{conj:hardness_output_probs}, this implies that the \textsf{PH} collapses to the third level.
\end{restatable}

Corollary~\ref{col:complexity_small_tvd_sampling} implies that sampling from the output distribution within an inverse polynomially small error is a classically hard problem unless the \textsf{PH} collapses.
We additionally note that the required $L_1$ distance can be even relaxed to a constant for $H_{3,4}$ if a constant fraction of $\mathbf{x} \in X$ (instead of $\mathbf{x} \in X_{n/2}$) satisfies the anticoncentration condition (see Appendix~\ref{app:hardness_sampling_proof} for details).

\begin{figure*}
    \centering
    \includegraphics[width=0.9\linewidth]{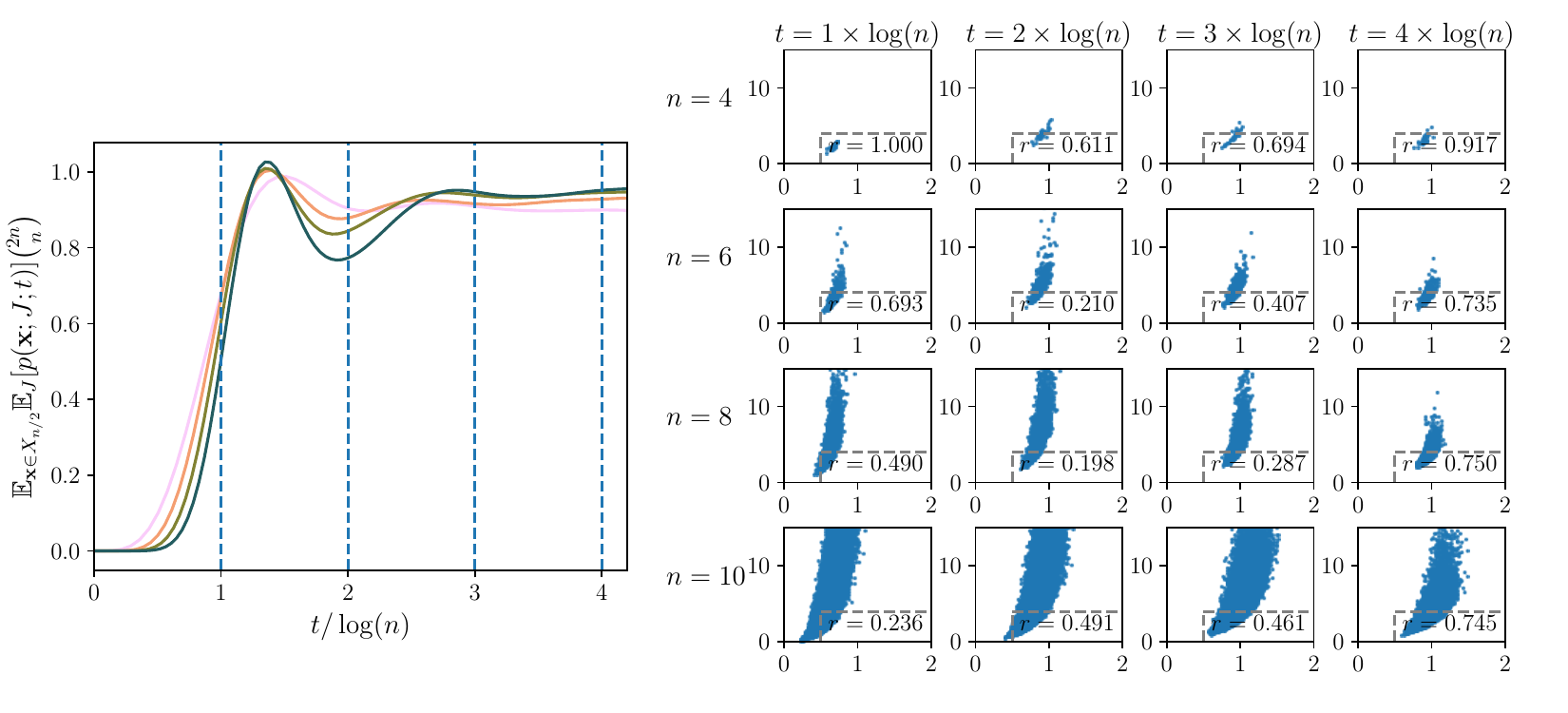}
    \caption{Left: Output probability $p(\mathbf{x};J;t)$ averaged over $J\sim \mathcal{N}(0,1)^{n \times n}$ and $\mathbf{x} \in X_{n/2}$ as a function of $t$.
    Results are plotted for $n=4$ (the lightest) to $n=10$ (the darkest) with a step size of $2$. The plot indicates that time to equilibrate $t_{\rm eq}$ scales with $\log n$. 
    Right: Scatter plots of $\mathbb{E}_J[p(\mathbf{x};J;t)] {2n \choose n}$ (x-axis) versus $\mathbb{E}_J[p(\mathbf{x};J;t)^2] {2n \choose n}^2$ (y-axis).
    For $t/\log n \in [1,2,3,4]$ and $n \in [4,6,8,10]$, we numerically compute $p(\mathbf{x};J;t)$ and $p(\mathbf{x};J;t)^2$ for all $\mathbf{x} \in X_{n/2}$ averaged over $J\sim \mathcal{N}(0,1)^{n \times n}$.
    Each data point represents $\mathbb{E}_J[p(\mathbf{x};J;t)] {2n \choose n}$ and $\mathbb{E}_J[p(\mathbf{x};J;t)^2] {2n \choose n}^2)$ for each $\mathbf{x}$.
    The dashed lines indicate $\mathbb{E}_J[p(\mathbf{x})] {2n \choose n} = 1/2$ and $\mathbb{E}_J[p(\mathbf{x})^2] {2n \choose n}^2 = 4$.
    We also present $r$ which is the ratio of $\mathbf{x} \in X_{n/2}$ satisfying $\mathbb{E}_J[p(\mathbf{x})] {2n \choose n} \geq 1/2$ and $\mathbb{E}_J[p(\mathbf{x})^2] {2n \choose n}^2 \leq 4$ (see main text).  
    Numerical results are obtained by fully diagonalizing the Hamiltonian for $n \in [4, 6, 8]$ and by using the second-order Trotter--Suzuki formula with a time step of $dt=10^{-6}$ for $N=10$.
    We selected such a small time step to ensure that the Trotter error is sufficiently smaller than $p(\mathbf{x};J;t)$ itself, which scales as $2^{-2n}$.
    The results are averaged over $2^{10}$ and $2^8$ random instances of $J\sim \mathcal{N}(0,1)^{n \times n}$ for $n \in [4,6,8]$ and $n=10$, respectively.
    }
    \label{fig:anticon_h3}
\end{figure*}

\section{Numerical verification of anticoncentration}~\label{sec:numerics_anticon}
In this section, we provide numerical evidence of the anticoncentration property given in Conjecture~\ref{conj:anticoncentration}.
We mainly focus on the XX model ($H_3$) here and provide results for the Ising model ($H_1$) in Appendix~\ref{app:anticon_ising}.

To numerically test anticoncentration, it is sufficient to compute
$\mathbb{E}_J[p(\mathbf{x};J;t)]$ and $\mathbb{E}_J[p(\mathbf{x};J;t)^2]$ for all $\mathbf{x} \in X_{n/2}$.
Let us suppose that $\mathbf{x}$ satisfies the following conditions:
\begin{align}
    \mathbb{E}_J[p(\mathbf{x};J;t)] &\geq K {2n \choose n}^{-1},\\
    \mathbb{E}_J[p(\mathbf{x};J;t)^2] &\leq \Lambda {2n \choose n}^{-2}, \label{eq:anticon_pz_eq}
\end{align}
where $K \leq 1$ and $\Lambda \geq 1$ are arbitrary constants.
Applying Paley-Zygmund inequality to $p(\mathbf{x};J;t)$ yields the following
\begin{align}
    &\underset{J \sim \mathcal{N}(0,1)^{n \times n}}{\mathrm{Pr}} \Bigl[ p(\mathbf{x};J;t) > \theta \mathbb{E}_J[p(\mathbf{x};J;t)] \Bigr] \nonumber \\
    &\qquad \geq (1-\theta)^2 \frac{\mathbb{E}_J[p(\mathbf{x};J;t)]^2}{\mathbb{E}_J[p(\mathbf{x};J;t)^2]},
\end{align}
which is satisfied for all $0 \leq \theta \leq 1$.
By entering $\theta=1/2$, we obtain
\begin{align}
    \underset{J \sim \mathcal{N}(0,1)^{n \times n}}{\mathrm{Pr}} \Bigl[ p(\mathbf{x};J;t) \geq \frac{K}{2} {2n \choose n}^{-1} \Bigr] \geq \frac{K^2}{4 \Lambda}.
\end{align}
Comparing the above equation to Conjecture~\ref{conj:anticoncentration}, we obtain $\alpha = K/2$ and $\beta = K^2/(4\Lambda)$.
Therefore, if the ratio of $\mathbf{x} \in X_{n/2}$ satisfying Eq.~\eqref{eq:anticon_pz_eq} converges to a constant as $n$ increases, we numerically confirm Conjecture~\ref{conj:anticoncentration}.
For numerical results, we choose $K=1/2$ and $\Lambda = 4$, yielding $\alpha = 1/4$ and $\beta = 1/64$.

For $n \in [4,6,8,10]$, we plot the output probability $p(\mathbf{x};J;t)$ averaged over $\mathbf{x} \in X_{n/2}$ and $J \sim \mathcal{N}(0,1)^{n \times n}$ in Fig.~\ref{fig:anticon_h3} (Left).
The results show that the averaged output probability reaches a stationary value for $t \gtrsim 3 \log n$ regardless of $n$.
This resembles a typical behavior of quantum equilibration~\cite{gogolin2016equilibration}.
However, compared to most existing literature on the equilibration of local observables, we observe the equilibration of the output probability $p(\mathbf{x};J;t)$.

We also present $p(\mathbf{x};J;t)$ and $p(\mathbf{x};J;t)^2$ averaged over $J$ in Fig.~\ref{fig:anticon_h3} (Right).
For each $n$ and $t$ presented in the plot, we also compute $r$, the ratio of $\mathbf{x} \in X_{n/2}$ which satisfies Eq.~\eqref{eq:anticon_pz_eq}.
Our results suggest that $r$ does not converge to zero as $n \rightarrow \infty$ for any of $t/\log n \in [2,3,4]$.
Specifically, we observe that $r$ oscillates with $n$ for $t/\log n \in [2,3]$, and $r \geq 0.7$ for $t = 4 \log n$ regardless of $n$.
For any value of considered $t$, $r$ is unlikely to converge to zero with $n$, thus suggesting the anticoncentration property by the Paley-Zygmund inequality.

We remark that the anticoncentration results we presented here rule out the possibility that $p(\mathbf{x};J;t)$ is close to $0$ for most of $\mathbf{x}$.
In this sense, anticoncentration is not merely an ingredient for the hardness proof but is also regarded as a signature of the average-case hardness to some extent (see Ref.~\cite{hangleiter2023computational} for further discussion).

\section{Quantum circuit realizations} \label{sec:quantum_circuit_realizations}
So far, we demonstrated that classical algorithms for computing the output distribution or sampling from it do not exist unless the \textsf{PH} collapses.
In this section, we show that a quantum computer can perform the sampling task in polynomial time using a quantum circuit implementation of the spin dynamics (i.e., $U=e^{-\complexi H t}$).

As some experimental platforms such as trapped ions~\cite{blatt2012quantum,monroe2021programmable} and neutral atoms~\cite{weimer2010rydberg,scholl2021quantum} naturally implement spin-$1/2$ Hamiltonians with noisy qubits,
it is tempting to implement our Hamiltonians directly on these systems.
However, an important obstacle is that our Hamiltonians are defined on a bipartite lattice, which makes a direct implementation challenging without adjusting the connectivity of the Hamiltonian.
While such connectivity-adjusted Hamiltonians may maintain classical hardness thanks to some degrees of long-range connectivity in those systems, it is an open question that needs separate proof. 
Noisy qubits are another limitation of using such native systems.
With the recent advances in simulating quantum systems using tensor networks, quantum advantage experiments with noisy qubits, once believed to be classically intractable, 
became amenable to classical simulations~\cite{pan2022solving,oh2023tensor}.

Hence, we discuss an implementation based on fault-tolerant intermediate-scale quantum (ISQ) hardware.
Most importantly, we prove that a circuit with a polynomial depth is sufficient to generate samples with a small $L_1$ distance,
which is classically challenging by Corollary~\ref{col:complexity_small_tvd_sampling}.

As the unitary operator generated by $H_1$, $e^{-\complexi H_1 t}=\prod_{ij} e^{-\complexi (tJ_{ij}/n) \sigma^{(i)}_x \tau^{(j)}_x}$, can be directly implemented in a quantum circuit with $n^2$ logical gates due to its commuting nature, we here focus on $H_{2,3,4}$ with non-commuting terms.
Still, we note that there is an interesting relation between the dynamics generated by $H_1$ and IQP circuits~\cite{bremner2016average,bremner2017achieving,bermejo2018architectures,haferkamp2020closing}.
They both can be implemented using commuting gates, but the underlying hardness problems are different: the permanent for $H_1$ and the universality of measurement-based quantum computation for IQP.
Thus, while we expect that most classical algorithms for IQP circuits can also simulate the dynamics of $H_1$, the required number of qubits to demonstrate quantum advantage counted based on the underlying problems can be much smaller for the dynamics of $H_1$ (see, e.g., Ref.~\cite{dalzell2020many} and our discussion below).

We now prove that a polynomial number of gates is sufficient to demonstrate quantum advantage for Hamiltonians $H_{2,3,4}$.
As a first step, we show that implementing a unitary $\tilde{U}$ that approximates $U=e^{-\complexi H t_0}$ within an inverse polynomially small additive error suffices to demonstrate quantum advantage.
For $\rho=\ket{\mathbf{y}_0} \bra{\mathbf{y}_0}$, the output probability distributions from two unitary operators are given as $p(\mathbf{x}) = \braket{\mathbf{x}|U\rho U^\dagger|\mathbf{x}}$ and $p'(\mathbf{x}) = \braket{\mathbf{x}|\tilde{U} \rho \tilde{U}^\dagger|\mathbf{x}}$.
Then, the $L_1$ distance between the two probability distributions is given by
\begin{align*}
    &\sum_{\mathbf{x}} |p(\mathbf{x}) - p'(\mathbf{x})| \\
    &= \sum_{\mathbf{x}} |\braket{\mathbf{x}|U\rho U^\dagger|\mathbf{x}} - \braket{\mathbf{x}| \tilde{U} \rho \tilde{U}^\dagger|\mathbf{x}}| \\
    &\leq 2 \Tr|U \rho U^\dagger - \tilde{U} \rho \tilde{U}^\dagger| \\
    &\leq 2 \Tr |U \rho (U^\dagger - \tilde{U}^\dagger)| + 2\Tr |(U - \tilde{U}) \rho \tilde{U}^\dagger| \\
    &\leq 4 \Vert U - \tilde{U} \Vert.
\end{align*}
For the first inequality, we used the property of the trace distance $2 \Tr|\rho - \rho'| = \max_{\{\Pi_i\}} \sum_i \Tr [\Pi_i(\rho - \rho')]$, in which the maximization is over all POVM $\{\Pi_i\}$.
The second and last inequalities are obtained from the triangular inequality and from H\"{o}lder's inequality $\Tr|AB| \leq \Vert A \Vert \Tr|B|$ with $\Tr|\rho U| = 1$ which is satisfied for any density matrix $\rho$ and any unitary operator $U$, respectively.

We now consider a quantum circuit implementation of $U=e^{-\complexi Ht_0}=(e^{-\complexi Ht_0/M})^M$ by using the Trotter decomposition~\cite{childs2021theory}.
By applying the $p$-th order Suzuki--Trotter formula, we obtain a quantum circuit implementation of $e^{-\complexi Ht_0/M}$ with an additive error of $O[\Upsilon_p (t/M)^{p+1}]$, where
\begin{align}
    \Upsilon_p = \sum_{i_1,\cdots, i_{p+1}=1}^\Gamma \Vert [h_{i_{p+1}},\cdots,[h_{i_2}, h_{i_1}] \Vert ,
\end{align}
with $H=\sum_{i=1}^\Gamma h_i$.
As we apply $M$ repetition of $e^{-\complexi Ht_0/M}$, the total error becomes $O[M\Upsilon_p (t_0/M)^{p+1}]$

For example, let us compute the error from the first order ($p=1$) Trotter decomposition using $H=H_3$. From the Suzuki--Trotter formula, we obtain
\begin{align*}
    &\exp \biggl[ -i \frac{t_0}{M}\sum_{ij} \frac{J_{ij}}{2n} \Bigl( \sigma^{(i)}_x \tau^{(j)}_x + \sigma^{(i)}_y \tau^{(j)}_y \Bigr) \biggr] = \\
    & \prod_{ij} \exp \biggl[ -i \frac{t_0}{M} \frac{J_{ij}}{2n} \Bigl( \sigma^{(i)}_x \tau^{(j)}_x + \sigma^{(i)}_y \tau^{(j)}_y \Bigr) \biggr] + O(\Upsilon_1 (t_0/M)^{-2}),
\end{align*}
where
\begin{align}
    &\Upsilon_1 = \sum_{i,j,k,l=1}^{n} \frac{|J_{ij}J_{kl}|}{4n^2} \times \nonumber \\
    &\qquad \Bigl\Vert \bigl[ \sigma^{(i)}_x \tau^{(j)}_x + \sigma^{(i)}_y \tau^{(j)}_y , \sigma^{(k)}_x \tau^{(l)}_x + \sigma^{(k)}_y \tau^{(l)}_y \bigr] \Bigr\Vert.
\end{align}
The commutator is nonzero only when either $i=k$ or $j=l$ (but not both); therefore, we have $2(n-1)$ possible choices of $(k,l)$ which contribute to the sum for each $(i,j)$.
As there exists a total of $n^2$ possibilities in selecting $(i,j)$, we have $\Upsilon_1 = O(n)$.
Thus by taking $M=\Theta(n^2 \log^2 n)$, we can bound the error by $O(\Upsilon_1 M^{-1} t_0^{2})=O(n^{-1})$ for $t_0=\Theta(\log n)$, which is sufficient to show the classical hardness by Corollary~\ref{col:complexity_small_tvd_sampling}.
In this case, the total number of gates is $Mn^2=\Theta(n^4 \log^2 n)$.

Similarly, the same type of calculation for the second-order Trotter formula results in $\Upsilon_2 = O(n)$, the total additive error of $O(nM^{-2} t_0^{3})$, and $2Mn^2$ total number of gates.
By choosing $M=\Theta(n \log^{3/2} n)$, we can bound the error by $O(n^{-1})$ using $\Theta(n^3 \log^{3/2} n)$ gates in total.
We can generalize this further by using the $p$-th order formula with arbitrary $p$; however, in any case, we can bound the error by $O(n^{-1})$ with a polynomial number of gates.

Finally, we roughly estimate how many logical qubits are required to demonstrate quantum advantage.
Recall that our problem's hardness relies on the permanent of submatrices, the sizes of which are determined by the output bitstrings.
As the output bitstrings are typically in $X_{m}$ with $m\approx n/2$, and computing the permanent of a dense matrix, the size of which is bigger than $\approx 50 \times 50$, is considered classically challenging (see, e.g., Refs.~\cite{Wu2018,Lundow2022}), we believe that an experiment with around $200$ qubits (i.e., $n \approx 100$) may be sufficient to demonstrate quantum advantage.
In addition, the number of gates to implement the Hamiltonian evolution for $n=100$ and $t_0=5 \times \log(100) \approx 23.03$ is approximately given by $1.2 \times 10^8$, $3.8 \times 10^8$, and $1.2 \times 10^9$, when the target Trotter error is $10^{-1}$, $10^{-2}$, and $10^{-3}$, respectively (see Appendix~\ref{app:estimating_number_of_gates} for details).

However, unlike \textsf{BosonSampling}, where the output probability distribution is proportional to the permanent of a matrix,
our main theorem only states that computing the output probability distribution implies computing the permanent; the converse is not necessarily true.
Thus, the output probability distribution of smaller systems can still be classically intractable for different reasons.
Indeed, 50-100 qubits are already regarded as extremely challenging for usual classical simulation algorithms for quantum spin dynamics (e.g., tensor network methods, Krylov subspace method, and time-dependent variational Monte-Carlo) unless the Hamiltonian has more structures that can be utilized (e.g., geometrically local connectivity).

\section{Conclusion and Outlook}
In this paper, we have proven the classical hardness of quantum spin dynamics.
We showed that the near-exactly computing the output distribution from the unitary operator generated by a wide class of many-body spin Hamiltonians is \textsf{\#P-hard}.
Those Hamiltonians include the Ising, XX, and Heisenberg models, defined on a bipartite lattice.
With the anticoncentration conjecture, according to which the output probability for each bitstring fluctuates only near its averaged value over the random Hamiltonian instances, we also proved that sampling from the output distribution is classically hard.
In contrast, the same task can be completed using a quantum circuit with a polynomial number of gates.

We expect that our techniques can be applied to other spin models.
Straightforward examples include the Ising, XX, and Heisenberg models on a bipartite lattice ($H_{1,2,3,4}$ considered in this study) with local fields in the $z$-direction (see Appendix~\ref{app:permanent_hamiltonian_moments}). 
Adding such local fields does not change $\braket{\mathbf{x}|H^l | \mathbf{y}_0}$ for all $k \leq n$, and the same complexity arguments hold.

Our results mainly focus on the hardness of computing the output distribution and the sampling task.
Sampling is an essential process for computing expectation values of many observables, e.g., the spin-$z$ correlation functions, using a quantum computer.
However, the classical hardness of estimating these observables must be studied separately; our results do not exclude the existence of an efficient algorithm for approximating $\braket{\psi_0 |e^{iHt} M e^{-iHt}|\psi_0}$, where $M$ is a local observable.
For example, some quantum circuits, including the \textsf{IQP}~\cite{ni2013commuting} and 2D short-depth circuits~\cite{bravyi2021classical}, allow efficient classical algorithms for estimating \text{local} observables, even though sampling from the output distribution is computationally hard in general.
We do not expect an efficient classical algorithm for estimating general observables for our Hamiltonians to exist, as our models are non-commuting and have more complex connectivity.
However, we leave it as an open question, given that a formal proof of hardness for computing observables will require different techniques.

Another open question is the verification scheme for the sampling task.
For a small system in which the time evolution can be classically computed, we can use the usual cross-entropy benchmarking~\cite{boixo2018characterizing}.
However, in contrast to instantaneous quantum polynomial circuits or random circuits, whose output distribution becomes the Porter--Thomas distribution, the distribution from our Hamiltonian dynamics does not necessarily follow it. 
Hence, a detailed verification scheme for a large number of qubits must be analyzed.

We provided a preliminary analysis of the requirements for an experiment demonstrating quantum advantage. These estimates are encouraging in suggesting that such experiments may be possible using a few hundred error-corrected qubits with less than a billion gates; these are within expectations of the capabilities of the first generation of fault-tolerant quantum computers. However, we emphasize that a more careful and detailed analysis is required to establish these numbers more confidently.

We close our paper by commenting on the generalization of our theory for going beyond the bipartite graph, which is essential for practical applications in chemistry, biology, and materials science.
We expect to draw similar conclusions for non-bipartite graphs, which generate hafnians or loop-hafnians instead of permanents.
Computing hafnians and loop-hafnians is \textsf{\#P-hard} as well (see, e.g., Ref.~\cite{bjorklund2019}), and we can repeat the proving processes as done in this paper.
The hardness results for general graphs would close the gap between our theory and quantum advantage in practical applications, for example, nuclear magnetic resonance spectroscopy~\cite{Sels2020}, where the complexity of estimating observables must be additionally considered.

\begin{acknowledgments}
CYP thanks Prof. Soonwon Choi for bringing Ref.~\cite{mark2023benchmarking} to our attention.
JH thanks Minhyeok Kang for insightful discussions. 
This work was partly supported by the Basic Science Research Program through the National Research Foundation of Korea (NRF), funded by the Ministry of Education, Science and Technology (NRF-2021M3E4A1038308, NRF-2021M3H3A1038085, NRF-2022M3H3A106307411, NRF-2023M3K5A1094805, and NRF-2023M3K5A109481311) and Institute for Information \& communications Technology Promotion (IITP) grant funded by the Korea government(MSIP) (No. 2019-0-00003, Research and Development of Core technologies for Programming, Running, Implementing and Validating of Fault-Tolerant Quantum Computing System).
This research used resources of the National Energy Research Scientific Computing Center, a DOE Office of Science User Facility supported by the Office of Science of the U.S. Department of Energy under Contract No. DE-AC02-05CH11231 using NERSC award NERSC DDR-ERCAP0025705.
JH acknowledges Xanadu for hosting his sabbatical year visit. 
The source code used for the simulations is available in the GitHub repository: \url{https://github.com/XanaduAI/hardness-spin-dynamics}.
\end{acknowledgments}


\bibliography{reference}

\appendix
\onecolumngrid

\section{Big O and related notations} \label{app:big_o_notation}
In this appendix, we summarize big $O$ and related notations.
Let $f(n)$ and $g(n)$ be functions of $n \in \mathbb{N}$.
Then, each notation is defined as follows:

\begin{itemize}
    \item $f(n) = o(g(n))$ iff for all $k > 0 $, there exists $n_0$ such that $|f(n)| < k g(n)$ for all $n > n_0$.
    \item $f(n) = O(g(n))$ iff there exist $k>0, n_0$ such that $|f(n)| \leq k g(n)$ for all $n > n_0$.
    \item $f(n) = \omega(g(n))$ iff for all $k>0$, there exists $n_0$ such that $f(n) > k g(n)$ for all $n > n_0$.
    \item $f(n) = \Omega(g(n))$ iff there exist $k>0, n_0$ such that $f(n) \geq k g(n)$ for all $n > n_0$.
    \item $f(n) = \Theta(g(n))$ iff $f(n) = O(g(n))$ and $f(n) = \Omega(g(n))$.
\end{itemize}

For example, if $f(n)=1/\log n$, for any given $k > 0$, $f(n) < k$ for all $n > e^k$. Therefore, $f(n) = o(1)$.

\section{Relating permanent of a matrix and moments of Hamiltonian}~\label{app:permanent_hamiltonian_moments}
In this appendix, we prove how the permanent of $J$ and its submatrices are obtained from moments of Hamiltonians.
We will use the results in this appendix extensively in the next appendix to prove the hardness of the output distribution.
First, we show the relation between the permanent and $H_3$.
Recall that $H_3$ is defined as
\begin{align}
    H_3 &= \sum_{i,j = 1}^n \frac{J_{ij}}{2n} \big[ \sigma^{(i)}_{x} \tau^{(j)}_{x} + \sigma^{(i)}_{y} \tau^{(j)}_{y} \big] = \sum_{i,j = 1}^n \frac{J_{ij}}{n} \big[ \sigma^{(i)}_{+} \tau^{(j)}_{-} + \sigma^{(i)}_{-} \tau^{(j)}_{+}\big],
\end{align}
where $\sigma^{(i)}_{\pm} = (\sigma^{(i)}_x \pm i\sigma^{(i)}_y)/2$ (and the same for $\tau^{(j)}_\pm$).

Following the notation in the main text, we consider binary vectors $\mathbf{x}$ and $\mathbf{y}$ of length $2n$,
and their substrings $\mathbf{x}^\sigma = x_1\cdots x_n$, $\mathbf{x}^\tau = x_{n+1}\cdots x_{2n}$ (the same for $\mathbf{y}_{\sigma,\tau}$).
Then, set $X_k$ is defined as $X_k = \{\mathbf{x} \in \mathbb{Z}_2^{2n}: \wt(\mathbf{x^\sigma} ) = k \text{ and } \wt(\mathbf{x^\tau} ) = n-k \}$.
Note that there is only a single vector in $X_0$, which is $\mathbf{y}_0 = \mathbf{0}^{\sigma}\mathbf{1}^{\tau}$.

We apply $H_3$ to $\ket{\mathbf{y}}$ in $\mathbf{y} \in X_k$.
In particular, note that each term of $H_3$ generates either a product state (with some abuse of notation, we use a product state and a bitstring interchangeably) in $X_{k+1}$ or $X_{k-1}$ with $J_{ij} \sigma^{(i)}_{+}\tau^{(j)}_{-}$ or with $J_{ij} \sigma^{(i)}_{-}\tau^{(j)}_{+}$, respectively,  from $\ket{\mathbf{y}}$.
In other words, we can write
\begin{align}
    H_3 \ket{\mathbf{y}} = \sum_{\mathbf{x} \in X_{k+1}|\mathbf{y}}  a_\mathbf{x} \ket{\mathbf{x}} + \sum_{\mathbf{x} \in X_{k-1}|\mathbf{y}}  b_\mathbf{x} \ket{\mathbf{x}} , \label{eq:h3_ket_y}
\end{align}
where $X_{k\pm 1}|\mathbf{y}$ denotes a bitstring in $X_{k\pm 1}$ generated form $\mathbf{y}$ (by applying $J_{ij} \sigma^{(i)}_{\pm}\tau^{(j)}_{\mp}$, respectively).
Given that $\mathbf{y}_0 \in X_0$, Eq.~\eqref{eq:h3_ket_y} implies that 
\begin{align}
    \braket{\mathbf{x} | H_3^l | \mathbf{y}_0} = 0 ,
\end{align}
for all $\mathbf{x} \in X_k$ when $l < k$.

Now, let us focus on $\braket{\mathbf{x} | H_3^k | \mathbf{y}_0}$ for $\mathbf{x} \in X_k$.
For convenience, let us define $\mathcal{S} = \{i \in [1, \cdots, n]: x_i = 1\}$, and $\mathcal{T} = \{i \in [1, \cdots, n]: x_{n+i} = 0\}$.
Thus, for $\mathbf{x} \in X_k$, $|\mathcal{S}|= |\mathcal{T}| = k$.
By explicitly writing $\braket{\mathbf{x} | H_3^k | \mathbf{y}_0}$, we obtain
\begin{align}
    &\braket{\mathbf{x}|H_3^k|\mathbf{y}_0} = 
    \bra{\mathbf{x}} \overbrace{\Bigl[ \sum_{i,j = 1}^n \frac{J_{ij}}{n} \big( \sigma^{(i)}_{+} \tau^{(j)}_{-} + \sigma^{(i)}_{-} \tau^{(j)}_{+}\big) \Bigr] \cdots \Bigl[ \sum_{i,j = 1}^n \frac{J_{ij}}{n} \big( \sigma^{(i)}_{+} \tau^{(j)}_{-} + \sigma^{(i)}_{-} \tau^{(j)}_{+}\big) \Bigr]}^{k \text{ times }} \ket{\mathbf{y}_0}.
\end{align}
Among all possible products, we want to extract terms with $\prod_{i \in S} \prod_{j \in T} \sigma^{(i)}_{+} \tau^{(j)}_{-}$.
As we apply $H_3$ for $k=|\mathcal{S}|= |\mathcal{T}|$ times, terms with only $\sigma^{(i)}_{+} \tau^{(j)}_{-}$ for $i \in \mathcal{S}$ and $j \in \mathcal{T}$ can contribute to this quantity.

In summary, we have
\begin{align}
    \braket{\mathbf{x}|H^k|\mathbf{y}_0} = \frac{1}{n^k} \sum_{p \in S_k} \sum_{p' \in S_k} \prod_{i \in p(\mathcal{S})} \prod_{j \in p'(\mathcal{T})} J_{ij} = \frac{k!}{n^k} \sum_{p \in S_k}\prod_{i=1}^k J_{\mathcal{S}[i], \mathcal{T}[p(i)]} = \frac{k!}{n^k} \mathrm{Per}(J_{\mathcal{ST}}). \label{eq:matrix_power_submatrix_per}
\end{align}
where $S_k$ is the symmetric group with $k$ elements and $p(\mathcal{S})$  for $p \in S_k$ is a permutation of $\mathcal{S}$ considering $\mathcal{S}$ as an ordered set. In addition, $J_{\mathcal{ST}}$ is a submatrix of $J$ only remaining rows in $\mathcal{S}$ and columns in $\mathcal{T}$.

We can extend this logic to $H=H_4$, which is defined as
\begin{align}
    H_4 &= \sum_{i,j = 1}^n \frac{J_{ij}}{n} \big[ \sigma^{(i)}_{+} \tau^{(j)}_{-} + \sigma^{(i)}_{-} \tau^{(j)}_{+} + \sigma^{(i)}_z \sigma^{(j)}_z \big].
\end{align}
The main difference to the previous Hamiltonian is that applying $H_4$ to $\ket{\mathbf{y}}$ for $\mathbf{y} \in X_{k}$ also generates terms in $X_k$, i.e.,
\begin{align}
    H_4 \ket{\mathbf{y}} = \sum_{\mathbf{x} \in X_{k+1}|\mathbf{y}}  a_\mathbf{x} \ket{\mathbf{x}} + \sum_{\mathbf{x} \in X_{k-1}|\mathbf{y}}  b_\mathbf{x} \ket{\mathbf{x}} + \sum_{\mathbf{x} \in X_{k}|\mathbf{y}}  c_\mathbf{x} \ket{\mathbf{x}}.
\end{align}
However, because the last term does not affect $\braket{\mathbf{x} | H_4^l | \mathbf{y}_0}$ for all $l < k$ when $\mathbf{x} \in X_k$, we still must apply $H_4$ at least $k$ times to obtain non-vanishing terms of $\mathbf{x} \in X_k$ from $\mathbf{y}_0 \in X_0$.

Similarly, we can compute $\braket{\mathbf{x} | H_4^k | \mathbf{y}_0}$.
By using the aforementioned argument, only coefficients for $\sigma^{(i)}_{+} \tau^{(j)}_{-}$ with $i \in \mathcal{S}, j \in \mathcal{T}$ contribute to this quantity. Therefore, we also find that
\begin{align}
    \braket{\mathbf{x} | H_4^k | \mathbf{y}_0} = \frac{k!}{n^k} \mathrm{Per}(J_{\mathcal{ST}}).
\end{align}

Let us now consider $H_1$, which can be reformulated as
\begin{align}
    H_1 = \sum_{i,j = 1}^n \frac{J_{ij}}{n} \sigma^{(i)}_{x} \tau^{(j)}_{x} = \sum_{i,j = 1}^n \frac{J_{ij}}{n} \Bigl[ \sigma^{(i)}_{+} \tau^{(j)}_{+} + \sigma^{(i)}_{+} \tau^{(j)}_{-} + \sigma^{(i)}_{-} \tau^{(j)}_{+} + \sigma^{(i)}_{-} \tau^{(j)}_{-} \Bigr].
\end{align}

Unlike previous Hamiltonians, applying $H_1$ to $\mathbf{y} \in X_{m}$ also creates terms not in $X:=\cup_i X_{i}$ but with the Hamming weights of $n+2$ and $n-2$.
Nevertheless, the overall results remain the same because these terms only go back to $X_m$ when another $H_1$ is applied. 
In other words, if we apply $H_1$ twice, we obtain chains, such as $X_k \rightarrow (n+2) \rightarrow X_k$ or $X_k \rightarrow (n+2) \rightarrow (n+4)$ where $(\cdot)$ indicates the Hamming weight of the resulting state.
Similarly, applying $H_1$ by $l$ times only creates terms in $X_{k-l}$ to $X_{k+l}$ and some other terms outside of $X$, i.e., 
\begin{align}
    H_1^l \ket{\mathbf{y}} = \sum_{\mathbf{x} \in X_{k-l} \cup \cdots \cup X_{k+l}|\mathbf{y}} a_{\mathbf{x}} \ket{\mathbf{x}} + \sum_{\mathbf{x} \notin \cup_{i} X_i} b_\mathbf{x} \ket{\mathbf{x}} ,
\end{align}
for $\mathbf{y} \in X_{k}$.
This implies $\braket{\mathbf{x} | H^l | \mathbf{y}_0}$ is still $0$ for all $l < k$ if $\mathbf{x} \in X_k$.

Again, for $\braket{\mathbf{x} | H_1^k | \mathbf{y}_0}$, only terms with $\sigma^{(i)}_+ \tau^{(i)}_{-}$ contribute to the result, which yields 
\begin{align}
    \braket{\mathbf{x} | H_1^k | \mathbf{y}_0} = \frac{k!}{n^k} \mathrm{Per}(J_{\mathcal{ST}}) ,
\end{align}
for $\mathbf{x} \in X_k$ as earlier examples. 
(See also Huh~\cite{Huh2022permanent}; the relation can be tested easily with symbolic computing packages such as Q3~\cite{Choi2022}, written in Wolfram language.)

Finally, we also obtain the same result for $H_2$, as $\sigma^{(i)}_z \tau^{(j)}_z$ terms cannot change any result for $\braket{\mathbf{x}|H^l|\mathbf{y}_0}$ with $l \leq k$ (when $\mathbf{x} \in X_k$).

To understand how the power of Hamiltonian is related to the permanent more easily, we provide an example of a puzzle involving two baskets $\sigma$ and $\tau$, with $n$ red and $n$ blue balls, each labeled from $1$ to $n$.
A constraint is that baskets $\sigma$ and $\tau$ can only contain red and blue balls, respectively.
There are four available operations:
\begin{itemize}
    \item A: Remove a blue ball $j$ from $\tau$ and add a red ball $i$ to $\sigma$ ($\sigma^{(i)}_{+} \tau^{(j)}_{-}$).
    \item B: Remove a red ball $i$ from $\sigma$ and add a blue ball $j$ to $\tau$ ($\sigma^{(i)}_{-} \tau^{(j)}_{+}$).
    \item C: Remove a red ball $i$ and a blue ball $j$ from each basket $\sigma$ and $\tau$ ($\sigma^{(i)}_{-} \tau^{(j)}_{-}$).
    \item D: Add a red ball $i$ and a blue ball $j$ to each basket $\sigma$ and $\tau$ ($\sigma^{(i)}_{+} \tau^{(j)}_{+}$).
\end{itemize}
We start from the initial state that no ball is in $\sigma$ and $n$ balls are in $\tau$.
If we can use exactly $k$ operations to ensure $k$ balls in $\sigma$ and $n-k$ balls in $\tau$, the only possible way is to apply A by $k$ times.
If we further require all balls with label $S$ be in $\sigma$ and all balls with $T$ be removed from $\tau$, the balls for operation A must be selected as $i \in \mathcal{S}$ for $\sigma$ and $j \in \mathcal{T}$ for $\tau$.
For each path from the initial state to the final state, the product of $J_{ij}$ for $(i,j)$ from each application of A contributes to $\braket{\mathbf{x} | H_1^k | \mathbf{y}_0}$.
As all permutations of $\mathcal{S}$ and $\mathcal{T}$ are allowed for each path, their summation is given by Eq.~\eqref{eq:matrix_power_submatrix_per}.

We remark that the same argument is also valid even when local fields in the $z$-direction are present, i.e., for
\begin{align}
    \tilde{H}_{1,2,3,4}=H_{1,2,3,4}+\sum_i h^{(1)}_i \sigma^{i}_z + \sum_j h^{(2)}_j \tau^{(j)}_z,
\end{align}
one can check that 
\begin{align}
    \braket{\mathbf{x} | \tilde{H}_{1,2,3,4}^l | \mathbf{y}_0} = 0 \text{ for all } l < k, \quad \braket{\mathbf{x} | \tilde{H}_{1,2,3,4}^k | \mathbf{y}_0} = \frac{k!}{n^k} \mathrm{Per}(J_{\mathcal{ST}}) ,
\end{align}
when $\mathbf{x} \in X_k$.
Thus, the results in the main text can also be applied to these Hamiltonians, implying that estimating output probabilities for these Hamiltonians is also classically hard (when $h_i^{(1)}=O(1)$ and $h_j^{(2)} = O(1)$).

\section{Complexity of the output distribution} \label{app:complexity_output_dist_proof}

This appendix details the proof of Theorem~\ref{thm:classical_hardness} discussed in the main text with additional exact hardness results.
We mainly consider $\mathbf{x}_0 = \mathbf{1}^{\sigma}\mathbf{0}^{\tau} \in X_{n}$ for hardness results, with some comments for general $\mathbf{x} \in X_m $ in Appendix~\ref{app:other_x_in_xm}.

\subsection{Polynomial expansion of $p(\mathbf{x}_0;J;t)$}
We prove several hardness results of estimating $p(\mathbf{x}_0;J;t)$ throughout this section.
Our main tool is a polynomial expansion of $p(t) = |\braket{\mathbf{x}_0| e^{-\complexi Ht}|\mathbf{y}_0}|^2$ (where we exclude $\mathbf{x}_0$ and $J$ from the arguments when they are obvious), which contains the permanent of $J$, i.e., $p(t)=\mathrm{Per}(J)^{2}t^{2n}/n^{2n}+\cdots$.
Let us define $f(t) = \braket{\mathbf{x}_0| e^{-\complexi Ht}|\mathbf{y}_0}$ the Taylor expansion of which is given by
\begin{align}
    f(t) = \sum_{k=0}^\infty a_k (-\complexi t)^k, \quad a_k = \frac{ \braket{\mathbf{x}_0| H^k|\mathbf{y}_0}}{k!} .
\end{align}

From the results in Appendix~\ref{app:permanent_hamiltonian_moments}, 
we know that $\braket{\mathbf{x}_0| H^k|\mathbf{y}_0} = 0$ for all $k < n$ and $\braket{\mathbf{x}_0| H^n|\mathbf{y}_0} = n! \mathrm{Per}(J)/n^n$ for any $H=H_{1,2,3,4}$ 
Therefore, the polynomial expansion of $p(t)=|f(t)|^2$ is given by
\begin{align}
    p(t) = \sum_{k=0}^{\infty} a_k t^k, \label{eq:p_t_expand}
\end{align}
where $a_k = 0$ for all $k < 2n$ and $a_{2n} = \mathrm{Per}(J)^2/n^{2n}$.
Thus, computing the coefficient for $t^{2n}$ of $p(t)$ (or equivalently, $p^{(2n)}(0)/(2n)!$ where $p^{(2n)}(t)$ is the $2n$-th derivative of $p(t)$) is classically difficult under Conjecture~\ref{conj:gaussian_permanent} introduced in the next subsection.

\subsection{Gaussian matrix permanent conjecture}
We introduce the Gaussian matrix permanent conjecture, which is central to our results.
\begin{conjecture}[Real-valued version of Theorem~7 in Ref.~\cite{Aaronson2011}]\label{conj:gaussian_permanent}
    For constants $\varepsilon,\delta \in (0, 1)$, estimating $\mathrm{Per}(J)^2$ within an additive error of $\varepsilon n!$ with a probability of at least $1-\delta$ over $J \sim \mathcal{N}(0,1)^{n \times n}$ is \emph{\textsf{\#P-Hard}}.
\end{conjecture}
Aaronson and Arkhipov~\cite{Aaronson2011} provided two main arguments why Conjecture~\ref{conj:gaussian_permanent} is plausible. First, they showed that estimating $\mathrm{Per}(J)^2$ within an additive error of $\varepsilon n!$ is polynomial-time equivalent to estimating $\mathrm{Per}(J)^{2}$ within a multiplicative error of $\varepsilon$ under the permanent anticoncentration conjecture (PAC).
Since estimating the permanent within a small multiplicative error is widely believed to be \textsf{\#P-hard}, and there are multiple pieces of evidence that the PAC is true, estimating the permanent within an additive error of $\varepsilon n!$ is also likely to be \textsf{\#P-hard}.
Second, they proved that estimating the permanent of Gaussian matrices within an additive error of $2^{-\mathrm{poly}(n)}$ with probability $3/4 + 1/\mathrm{poly}(n)$ 
is indeed a \textsf{\#P-hard} problem. See also Appendix~\ref{app:approximate_gaussian_permanent} for further discussion.

\subsection{Hardness of exactly computing the output distribution for exponentially short time}
Conjecture~\ref{conj:gaussian_permanent} readily implies that the exact evaluation of $p(t)$ at $t \leq n^{-cn}$ with any $c > 1/2$ is classically difficult. 
Herein, we prove this fact.

Let us recall the polynomial expansion of $f(t):=\braket{\mathbf{x}_0 | e^{-\complexi Ht} | \mathbf{y}_0}$.
By Taylor's theorem, there exists $0 < \tau_0 < t$ such that
\begin{align}
    f(t) = \frac{\mathrm{Per}(J)}{n^n} (-\complexi t)^{n} + \frac{f^{(n+1)}(\tau_0)}{(n+1)!} t^{n+1},
\end{align}
where
\begin{align}
    f^{(n+1)}(\tau_0) = (-\complexi)^{n+1} \braket{\mathbf{x}_0 | H^{n+1} e^{-\complexi H\tau_0}| \mathbf{y}_0}.
\end{align}

Therefore, $p(t) = |f(t)|^2$ becomes
\begin{align}
    p(t) = \frac{t^{2n}}{n^{2n}}\Bigl| \mathrm{Per}(J) + \xi(t) \Bigr|^2
\end{align}
where
\begin{align}
    \xi(t) = -\complexi\frac{n^n \braket{\mathbf{x}_0 | H^{n+1} e^{-\complexi H\tau_0}| \mathbf{y}_0}}{(n+1)!} t,
\end{align}
which results in
\begin{align}
    |\xi(t)| \leq \frac{\Vert H \Vert^{n+1} n^n}{(n+1)!} t.
\end{align}

Thus, the exact computation of $p(t)$ implies that we can compute
\begin{align}
    \mathrm{Per}(J)^2 + 2\mathcal{R}[\xi(t)] \mathrm{Per}(J) + |\xi(t)|^2,
\end{align}
where $\mathcal{R}[\xi(t)]$ is the real part of $\xi(t)$.
This implies that we can estimate $\mathrm{Per}(J)^2$ within an additive error of $\epsilon n!$ when $|\xi(t)| = o(\sqrt{n!})$ and $|\mathrm{Per}(J)| =O(\sqrt{n!})$.

From Chebyshev's inequality, we obtain
\begin{align}
    \underset{J \sim \mathcal{N}(0,1)^{n \times n}}{\mathrm{Pr}}\Bigl[ |\mathrm{Per}(J)| < k \sqrt{n!} \Bigr] \geq 1 - \frac{1}{k^2},
\end{align}
given that the variance of $\mathrm{Per}(J)$ for Gaussian random matrices is $n!$~\cite{Aaronson2011}.
Therefore, by choosing, e.g., $k=10$, we can estimate the permanent within $o(n!)$ with probability $\geq 0.99$ over random Gaussian matrices $J$ when $|\xi(t)|^{2} = o(n!)$, which is a \textsf{\#P-hard} problem by Conjecture~\ref{conj:gaussian_permanent}.

By Stirling's approximation, for $\Vert H \Vert \leq 3n$, we obtain
\begin{align}
    |\xi(t)| \leq (3 e)^{n+1} n^n t,
\end{align}
which becomes $o(\sqrt{n!})$ when $t = n^{-c n}$ with $c > 1/2$. 
As $\Vert H \Vert \leq 3n$ is satisfied for most of $H$ (see Lemma~\ref{lm:bounding_prob_sum_of_gaussian} in Appendix~\ref{app:proof_details}), it is \textsf{\#P-hard} to exactly compute $p(t)$ for small $t$.

We note that a similar result was obtained for the Ising model in Ref.~\cite{Fefferman2017}.
However, the exact computation considered in this section and Ref.~\cite{Fefferman2017} is impossible for any realistic computational model as a quantity is subject to an error introduced by finite-precision arithmetic~\cite{Minsky1967Computation}.
Moreover, if we consider sampling tasks (see Sec.~\ref{sec:hardness_sampling}), $t$ must be sufficiently large (at least $\Omega(\log n)$) to make the output distribution anticoncentrate.
Thus, we must consider how the error introduced in the estimation of $p(t)$ for larger $t$ affects our estimation of $\mathrm{Per}(J)^2$.

\subsection{Average-case hardness of exactly computing the truncated output distributions}
We now consider the average-case hardness with respect to $t$.
Based on the findings in Ref.~\cite{bouland2019complexity}, we consider a classical oracle that exactly computes a truncated approximation of $p(t) = |\braket{\mathbf{x}_0 | e^{-\complexi Ht} | \mathbf{y}_0}|^2$.
In other words, an oracle returns an exact value of $p_K(t)$ given as 
\begin{align}
    p_K(t) = \sum_{k=0}^{K} \frac{p^{(k)}(0)}{k!} t^k = \frac{\mathrm{Per}(J)^2}{n^{2n}} t^{2n} + \cdots + \frac{p^{(K)}(0)}{K!}t^K ,
\end{align}
with $K \geq 2n$. Here, $p^{(K)}(t)$ is the $K$-th derivative of $p(t)$.
Then, we can prove the hardness of exactly computing $p_K(t)$ by using the Berlekamp--Welch algorithm.
\begin{theorem}[Berlekamp--Welch algorithm~\cite{berlekamp_welch}]
   Let $q(t)$ be a degree $d$ polynomial. Suppose a dataset $\{t_i, y_i\}_{i=1}^L$ with all distinct $\{t_i\}$ with the promise that $y_i = q(t_i)$ at least $\min(d+1, (L+d)/2)$ points. Then, $q(t)$ can be exactly recovered in $\mathrm{poly}(L,d)$ deterministic time.
\end{theorem}

Now, suppose that there exists an oracle that exactly computes $p_K(t)$ over $t \in [a,b]$ for any $b-a > 0$ with probability $>1/2 + 1/\mathrm{poly}(n)$.
By the Chernoff bound, we can find $L=\mathrm{poly}(n)$ such that at least $(L+d)/2$ points are correct with a success probability of $1-e^{-\Theta(n)}$.
Consequently, we can recover $p_K(t)$ by the theorem.

Given that we recovered $p_K(t)$ exactly, all coefficients of $p_K(t)$ are also known (in fact, the Berlekamp-Welch algorithm returns all the coefficients).
By extracting the coefficient for $t^{2n}$, we can solve the square of the permanent of $J$, which is a \textsf{\#P-hard} problem by Conjecture~\ref{conj:gaussian_permanent}.
This result can be summarized by the following theorem.

\begin{theorem}\label{thm:output_prob_exact_hardness}
    For any $K \geq 2n$, it is \emph{\textsf{\#P-hard}} to exactly compute $p_K(t)$ with probability $1/2 + 1/\mathrm{poly}(n)$ over choices of $t \in [a, b]$ with $a<b$.
\end{theorem}

\subsection{Approximate hardness from robust polynomial regression}
Theorem~\ref{thm:output_prob_exact_hardness} only proves the average-case hardness for an \textit{exact} oracle returning $p_K(t)$.
Given that our function $p(t)$ is not an exact polynomial function, our goal is to find a polynomial that approximates $p(t)$ assuming that data points $\{t_i, p(t_i)\}$ can be obtained from the oracle.
An approach suggested in Refs.~\cite{Aaronson2011,bouland2019complexity} is to combine Paturi's Lemma~\cite{paturi1992degree} and Rakhmanov's result~\cite{rakhmanov2007bounds}.
In this approach, the error from the estimation of $p(t)$ obtains an additional factor, which is exponentially large in the degree of the polynomial. It can be compensated by requiring exponentially small estimation errors.

Parturi's lemma and Rakhmanov's results are presented as follows.
\begin{lemma}[Parturi~\cite{paturi1992degree}]
    Let $q(t)$ be a polynomial of degree $d$, and suppose $|q(t)| \leq \delta$ for all $|t| \leq \Delta$. Then $q(1) \leq \delta \exp(2d (1 + \Delta^{-1}))$.
\end{lemma}
\begin{lemma}[Simplified version of Rakhmanov~\cite{rakhmanov2007bounds}]
    Let $q(t)$ be a polynomial of degree $d$, and suppose that $|q(t_i)| \leq \delta$ for equidistant points $\{t_i\}_{i=1}^L$ in $[-1, 1]$ with $L > d$. Then $|q(t)| \leq O(\epsilon)$ for all $|t| \leq 1$.
\end{lemma}

By combining those lemmas, it is possible to prove the average-case approximate hardness of estimating the permanent of Gaussian random matrices~\cite{Aaronson2011} (albeit the same technique does not work for RCS; see Appendix~\ref{app:difficulty_fp_algo_rakh}).

This type of technique, which provides error bounds between a polynomial we want to know (ground truth) and the polynomial reconstructed from noisy samples, is called \textit{robust polynomial regression}.
By utilizing recent progress in this field~\cite{movassagh2023hardness,sherstov2012making,kane2017robust,bouland2022noise,krovi2022average,kondo2022quantum}, we can prove a similar result for estimating high-order derivatives as follows.

\begin{lemma}\label{lm:p_poly_coefficient_bound}
    Let $q(t)$ be a degree-$d$ polynomial, and suppose that a dataset $\{(t_i, y_i)\}_{i=1}^L$ is given, in which $\{t_i\}$ are $L=d+1$ equidistant points in $[t_0(1-\Delta), t_0(1+\Delta)]$ and $|y_i-q(t_i)| \leq \delta$ for all $i$, with  
    $0<\Delta<1$.
    Then, the coefficient of $q(t)$ for $t^k$ can be estimated in $\mathrm{poly}(d)$ time deterministically within an error of $O[\delta t_0^{-k} (4/\Delta)^{d} {d \choose k}]$.
\end{lemma}
We prove Lemma~\ref{lm:p_poly_coefficient_bound} in Appendix~\ref{app:proof_details} by constructing a degree-$d$ polynomial using the Lagrange interpolation and by bounding the difference between the constructed polynomial and $q(t)$.

We are ready to prove a primitive version of Theorem~\ref{thm:classical_hardness}.
\begin{theorem}[FP algorithm]\label{thm:fp_algorithm}
    Let $t_0=O(\log n)$ and $0 < \Delta < 1$ be a constant. 
    Under Conjecture~\ref{conj:gaussian_permanent}, there exists a constant $c_0>0$ such that the following statement holds: If there exists an efficient oracle that evaluates $p(\mathbf{x};J;t)$ for $\mathbf{x} \in X_m$ with $m \geq \sqrt{n}$ and $t \in [t_0(1-\Delta), t_0(1+\Delta)]$ within an additive error of $n^{-c_0n}$ with a probability of $\geq 1-1/n^2$ over $t$ and a constant probability of $\eta>0$ over $J \sim \mathcal{N}(0,1)^{n \times n}$, then $\mathsf{FP}=\mathsf{\#P}$.
\end{theorem}

We prove Theorem~\ref{thm:fp_algorithm} for $\mathbf{x} = \mathbf{x}_0 \in X_n$ in the rest of this subsection (see Appendix~\ref{app:other_x_in_xm} for other $\mathbf{x} \in X_m$).
We again use $p_K(t)$, a truncated Taylor series of $p(t)$ with degree $K$, to estimate how the error from the estimation of $p(t)$ (where we exclude $\mathbf{x}_0$ and $J$ from the arguments as when they are obvious) affects our estimation of $p^{(2n)}(0)/(2n)!$. 
Let us define $p_K(t)$ such that
\begin{align}
    p_K(t) = \sum_{k=2n}^K a_k t^k,
\end{align}
where $a_{k} = p^{(k)}(0)/k!$ with a special case of $a_{2n}=\mathrm{Per}(J)^2/n^{2n}$.
Using Taylor's theorem, we have
\begin{align}
    |p(t) - p_K(t)| \leq \frac{|p^{(K+1)}(\tau)|}{(K+1)!} t^{K+1} ,
\end{align}
for $0 \leq \tau \leq t$.
With $p(t) = f(t)f(t)^{*}$, where $f(t) = \braket{\mathbf{x}_0| e^{-\complexi Ht}|\mathbf{y}_0}$, we can expand $p^{(K+1)}(\tau)$ as follows:
\begin{align}
    p^{(K+1)}(\tau) = \sum_{l=0}^{K+1} {K+1 \choose l} f^{(l)}(\tau) [f^{(K+1-l)}(\tau)]^*.
\end{align}
Using an upper bound of $f^{(l)}(\tau)$ as given in below 
\begin{align}
    |f^{(l)}(\tau)| = |\braket{\mathbf{x_0}|H^{l} e^{-\complexi H\tau_0}|\mathbf{y}_0}| \leq \Vert H \Vert^l,
\end{align}
where $\Vert H \Vert$ is the operator norm of $H$,
we can bound $p^{(K+1)}(\tau)$ as
\begin{align}
    |p^{(K+1)}(\tau)| \leq \sum_{l=0}^{K+1} {K+1 \choose l} \Vert H \Vert^{K+1} \leq (2 \Vert H \Vert)^{K+1}.
\end{align}
Finally, we have the truncation error bound of probability, denoted $\epsilon_K$ as follows:
\begin{align}
    |p(t) - p_K(t)| \leq \frac{(2 \Vert H \Vert t)^{K+1}}{(K+1)!} := \epsilon_K \label{eq:epsilon_K_def} .
\end{align}
Thus, by estimating $p(t)$ within an additive error of $\epsilon$, we can also estimate $p_K(t)$ within an additive error of $\epsilon' := \epsilon + \epsilon_K$.

In summary, the problem reduces to the estimation of $2n$-th coefficient of $p_K(t)$ (i.e., for $t^{2n}$), wherein noisy data for $p(t)$ near $t=t_0$ are provided.
Then, the theorem is proved as follows:
(1) For a given oracle that approximates $p(\mathbf{x}_0;J;t)$, we call the oracle $K+1$ times for a given $J$ to generate an equidistant dataset.
The oracle successfully returns all data points with probability $\geq 1-(K+1)/n^2$.
(2) We use Lemma~\ref{lm:p_poly_coefficient_bound} to estimate $a_{2n}$, which will be our estimation of $\mathrm{Per}(J)^2/n^{2n}$.
Thus, the final error for $\mathrm{Per}(J)^2$ is given by $O[(nt_0^{-1})^{2n}\epsilon'(4/\Delta)^{K} {K \choose 2n}]$, where $n^{2n}$ comes from the denominator of $a_{2n}$, and by identifying $\delta=\epsilon'$, $k=2n$, $d=K$ in Lemma~\ref{lm:p_poly_coefficient_bound}.
(3) By bounding this error to $o(n!)$, we can prove that the existence of such an oracle solves a \textsf{\#P-hard} problem by Conjecture~\ref{conj:gaussian_permanent}.

For the proof, we restrict $H$ to have $\Vert H \Vert \leq \kappa n$ with $\kappa=3$.
Here, $\kappa$ can be an arbitrary constant that sufficiently reduces the outlier probability, i.e., $\mathrm{Pr}[\Vert H \Vert > \kappa n ] = o(1)$. 
This adds to the failure probability of our algorithm.
In Appendix~\ref{app:proof_details}, we show that $H$ constructed using a random Gaussian matrix $J$ produces a small outlier probability given by $e^{-\Theta(n^2)}$ when $\kappa=3$ for all $H=H_{1,2,3,4}$.

For those $H$, we show that $K=\alpha n \log n$ estimates $\mathrm{Per}(J)^2$ within an additive error of $o(n!)$ when $\epsilon=1/n^{c_0 n}$ for a constant $c_0$.
First, let us bound $\epsilon_K$ [Eq.~\eqref{eq:epsilon_K_def}].
As $t_0 = O(\log n)$, we can find $C>0$ such that $t_0 \leq C \log n$ for a sufficiently large $n$. Then, we obtain
\begin{align}
    &\epsilon_K \leq \frac{[6 n t_0 (1+\Delta)]^{K+1}}{(K+1)!} \leq \frac{(12 n t_0)^{K+1}}{(K+1)!}  \nonumber \\
    &\quad \leq \Bigl( \frac{12 e n t_0}{K+1} \Bigr)^{K+1} \leq \Bigl( \frac{12 e C}{\alpha} \Bigr)^{\alpha n \log n + 1} ,
\end{align}
where we use $0 < \Delta < 1$ for the second inequality and Stirling's approximation for the third inequality. 
In addition, if we choose $c_0 \geq -\alpha \log(12eC/\alpha)$ for $\epsilon=n^{-c_0 n}$, we obtain
\begin{align}
    \epsilon'=\epsilon + \epsilon_K = O\biggl[ \Bigl( \frac{12 e C}{\alpha} \Bigr)^{\alpha n \log n } \biggr] .
\end{align}

Next, we extract the error factor from Lemma~\ref{lm:p_poly_coefficient_bound}, by identifying $q(t)=p_K(t)$ with $d= K = n \log n$ and $k=2n$.
From the entropic bound of the binomial coefficient given by
\begin{align}
    {K \choose 2n} \leq 2^{K \mathcal{H}(2n/K)},
\end{align}
where $\mathcal{H}(p) = -p \log_2(p) - (1-p) \log_2 (1-p)$ is the Shannon entropy, it follows that ${K \choose 2n} \leq 1.5^{K}$ for sufficiently large $n$. 
Thus, the error factor from Lemma~\ref{lm:p_poly_coefficient_bound} becomes
\begin{align}
    O\Bigl[ \epsilon'  t_0^{-2n} \Bigl( \frac{4}{\Delta} \Bigr)^{K} {K \choose 2n} \Bigr] &= O\Bigl[\epsilon'  t_0^{-2n} \Bigl( \frac{6}{\Delta} \Bigr)^{\alpha n \log n}\Bigr] = O\biggl[ \Bigl( \frac{72 e C}{\alpha \Delta} \Bigr)^{\alpha n \log n } \biggr].
\end{align}
By multiplying the remaining factor $n^{2n}$ from the denominator of $a_{2n}$ with the last expression, we obtain the overall error for $\mathrm{Per}(J)^2$ by $O(n^{n[2 + \alpha \log(72eC/\alpha \Delta)]})$.
This can be further bounded by $O(n^{\mu n}) = o(n!)$ with $\mu < 1$ as we can find $\alpha$ such that $\alpha \log (72eC/\alpha \Delta) < -1$.
Our result is summarized as follows: Suppose there exists a classically efficient oracle ($\in \mathsf{FP}$) that can estimate $p(t)$ within an additive error of $n^{-c_0 n}$ for $t \in [t_0(1-\Delta), t_0(1+\Delta)]$ with a constant probability $\eta > 0$ over $J \sim \mathcal{N}(0,1)^{n \times n}$.
Then, we can estimate $\mathrm{Per}(J)^2$ within an additive error of $o(n!)$ by calling the oracle $\mathrm{poly}(n)$ times.
The success probability of the algorithm over $J\sim \mathcal{N}(0,1)^{n \times n}$ is given by $\eta-e^{-\Theta(n^2)} - (\alpha n \log n+1)/n^2 \geq 7\eta/8$ for sufficiently large $n$.
Then, it follows that \textsf{FP=\#P} by Conjecture~\ref{conj:gaussian_permanent}, which also implies \textsf{P=NP}~\cite{arora2009computational}.

However, Theorem~\ref{thm:fp_algorithm} required an oracle that can estimate $p(\mathbf{x};J;t)$ for a range of $t \in [t_0(1-\Delta), t_0(1+\Delta)]$.
Thus, such an oracle is much stronger than we want for Conjecture~\ref{conj:hardness_output_probs}, which estimates $p(\mathbf{x};J;t)$ only at $t=t_0$.
Fortunately, it is also possible to prove the hardness using an oracle for $t=t_0$ if we require smaller additive errors in the estimation and use a more powerful algorithm.

\subsection{Proof of Theorem~\ref{thm:classical_hardness} for $\mathbf{x} \in X_n $}
In this subsection, we prove Theorem~\ref{thm:classical_hardness} for $\mathbf{x} = \mathbf{x}_0 \in X_n$ by providing an \textsf{FBPP} algorithm that estimates $\mathrm{Per}(J)^2$ within a sufficiently small error when an oracle that evaluates $p(\mathbf{x}_0;J;t)$ at $t=t_0$ is provided.
One of the main differences of the algorithm presented in this subsection compared to the one introduced in Theorem~\ref{thm:fp_algorithm} is that the oracle only requires to return the estimated values of $p(t)$ at $t=t_0$ instead of that of the range $[t_0(1-\Delta), t_0(1+\Delta)]$.
This is because we can estimate $p(t)$ near $t_0$ if the estimation of $p(t_0)$ is possible for Gaussian random matrices $J$, which follows from the lemma below.

\begin{lemma} \label{lm:tvd_gaussian}
    An oracle that estimates $p(J;t_0)$ within an additive error of $\epsilon$ with probability $\eta$ over $J\sim \mathcal{N}(0,1)^{n \times n}$ can estimate $p(J;t)$ within the same error with a probability of at least $\eta - (3/2)n\sqrt{|(t/t_0)^2 - 1|}$.
\end{lemma}
\begin{proof}
    Let $\mathcal{O}(J)$ be an oracle that satisfies
    \begin{align}
        \mathrm{Pr}_{J \sim \mathcal{N}(0,1)}[|\mathcal{O}(J) - p(J;t_0)| \leq \epsilon] > \eta .
    \end{align}
    Let $E$ be the events where $|\mathcal{O}(J) - p(J;t_0)| \leq \epsilon$.
    As the distribution of $p(J;t)$ for $J \sim \mathcal{N}(0,1)$ is the same as that of $p(J;t_0)$ with $J \sim \mathcal{N}(0,(t/t_0)^2)$, the difference between the probability of having $E$ is upper bounded by the total variational distance, i.e.,
    \begin{align}
        |\mathrm{Pr}_{J \sim \mathcal{N}(0,(t/t_0)^2)}[E] - \mathrm{Pr}_{J \sim \mathcal{N}(0,1)}[E] | \leq \delta ,
    \end{align}
    where $\delta$ is the total variational distance between two Gaussian distributions. 
    This completes the proof as $\delta \leq (3/2) n \sqrt{|(t/t_0)^2-1|}$~\cite{devroye2018total}.
\end{proof}

The lemma implies that, when $\eta \geq 3/4$ is a constant, by choosing $\Delta \leq \beta / n^2$ with $\beta = 1/64$, the oracle can estimate $p(t)$ within an error $\epsilon$ with probability $> 1/2$ over $t \in [t_0(1-\Delta),t_0(1-\Delta)]$.
Then, we can apply the following lemma to estimate the $2n$-th coefficient of $\tilde{q}(t)$, i.e., $a_{2n}$.

\begin{lemma}\label{lm:bpp_poly_coefficient_bound}
    For a degree-$d$ polynomial $q(t)$, suppose that there exists an efficient oracle that estimates $q(t)$ within error $\delta$ with a constant probability larger than $1/2$ over $t \in [t_0(1-\Delta), t_0(1+\Delta)]$.
    Then there exists an \emph{\textsf{FBPP}} algorithm that returns the coefficient of $q(t)$ for $t^k$, within an additive error of $O[ \delta t_0^{-k} (4/\Delta)^d {d \choose k}]$ with a success probability of at least $2/3$.
\end{lemma}
See Appendix~\ref{app:proof_details} for the proof. We now prove Theorem~\ref{thm:classical_hardness} for $\mathbf{x} = \mathbf{x}_0$, which we recall as follows.

\ThmClassicalHardness*

We follow the same steps for the proof for Theorem~\ref{thm:fp_algorithm} to find an error bound for $\mathrm{Per}(J)^2$.
Let us bound $\epsilon_K$ for $K= \alpha n^3 \log n$ and $t_0 \leq C \log n$, which is given by
\begin{align}
    &\epsilon_K \leq \Bigl( \frac{12e n t_0}{K+1} \Bigr)^{K+1} \leq  \Bigl( \frac{12e n t_0}{K} \Bigr)^{K+1} \nonumber \\
    &= \Bigl( \frac{12 e C}{\alpha n^2}\Bigr) \Bigl( \frac{12e C}{\alpha} \Bigr)^{\alpha n^3 \log n} n^{-2\alpha n^3 \log n} \nonumber \\
    &\leq \Bigl( \frac{12e C}{\alpha} \Bigr)^{\alpha n^3 \log n} n^{-2\alpha n^3 \log n },
\end{align}
where the last inequality is satisfied for $n \geq (12eC / \alpha)^{1/2}$.
Thus for any $c_1 > 2\alpha $ and $\epsilon = n^{-c_1 n^3 \log n}$, we have $\epsilon = o(\epsilon_K)$, i.e.,
\begin{align}
    \epsilon' = \epsilon + \epsilon_K = O \biggl[ \Bigl( \frac{12e t_0}{\alpha} \Bigr)^{\alpha n^3 \log n} n^{-2\alpha n^3 \log n} \biggr].
\end{align}

We now apply Lemma~\ref{lm:bpp_poly_coefficient_bound} by putting $\delta=\epsilon'$, $\Delta \leq \beta/n^2$, $d=K$, and $k=2n$.
This yields the following additive error bound:
\begin{align}
    \epsilon' t_0^{-2n} \left( \frac{4 n^2 }{ \beta} \right)^{K} {K \choose 2n} \leq \epsilon' t_0^{-2n} \left(\frac{6}{\beta}\right)^{\alpha n^3 \log n} n^{2\alpha n^3 \log n}.
\end{align}
Consequently, the permanent can be estimated within an additive error of
\begin{align}
    O\Biggl[ \Bigl( \frac{n}{t_0} \Bigr)^{2n} \Bigl( \frac{72 et_0}{\alpha \beta} \Bigr)^{\alpha n^3 \log n} \Biggr],
\end{align}
which is $o(n!)$ for $\alpha > (72 et_0)/\beta$.
Hence, the existence of an efficient classical oracle ($\in \mathsf{FP}$) for estimating $p(t_0)$ within an additive error of $n^{-c_1 n^3 \log n}$ with probability $\geq 3/4$ over $J \sim \mathcal{N}(0,1)^{n \times n}$ implies that there exists an \textsf{FBPP} algorithm for estimating the permanent of Gaussian matrices; this also implies that \textsf{FBPP=\#P}.

\subsection{Difficulty of combining Lemmas~\ref{lm:p_poly_coefficient_bound} and \ref{lm:tvd_gaussian}} \label{app:difficulty_fp_algo_rakh}
Given Lemma~\ref{lm:tvd_gaussian}, one might ask whether we can combine Lemmas~\ref{lm:p_poly_coefficient_bound} and \ref{lm:tvd_gaussian} to obtain a \textit{deterministic} algorithm (i.e., not \textsf{FBPP} but \textsf{FP}) for estimating the $k$-th coefficient of $q(t)$ only using an oracle for $p(\mathbf{x};J;t_0)$.
However, the main problem of combining these results is that Lemma~\ref{lm:p_poly_coefficient_bound} requires $L>d$ data points \textit{a priori}.
They are only accessible when the oracle can output $p(t)$ with a probability greater than $1-1/L$ over $t \in [t_0(1-\Delta), t_0(1+\Delta)]$,
and this is the main reason why we cannot easily combine Lemmas~\ref{lm:p_poly_coefficient_bound} and \ref{lm:tvd_gaussian}.

To understand the situation better, let us recall the setting for Theorem~\ref{thm:fp_algorithm}.
The error from the polynomial approximation (the difference between $p_K(t)$ and $p(t)$) was given by $\epsilon_K = ( 12net_0/K )^{K}$
which is magnified by a factor of $t_0^{-2n}(4/\Delta)^K 2^{n \mathcal{H}(2n/K)}\leq t_0^{-2n}(6/\Delta)^K$ for the derivatives.
Thus, the overall error becomes
\begin{align}
    t_0^{-2n} \Bigl( \frac{72 net_0}{K \Delta} \Bigr)^K.\label{eq:poly_approx_factor_magnified}
\end{align}

Now let us assume that there exists an oracle that computes $p(t_0)$ with probability $\eta > 0$ over $J \sim \mathcal{N}(0,1)^{n \times n}$.
By Lemma~\ref{lm:tvd_gaussian}, we can estimate $p(t)$ for $t \in [t_0(1-\Delta),t_0(1+\Delta)]$ with probability $\eta - (3/2)n \sqrt{ [(1+\Delta)^2-1]}$.
However, the condition requiring $L$ data points in advance is only satisfied when the oracle outputs a correct estimation of $p(t)$ for a given $t_i$ with a probability greater than $1-O(1/K)$.
To satisfy this condition, we need $\eta = 1-O(1/K)$ and $\Delta = O(1/(n^2 K^2))$.
By entering this $\Delta$ into Eq.~\eqref{eq:poly_approx_factor_magnified}, we have 
\begin{align}
    O[t_0^{K -2n} (72e n^3 K)^K],
\end{align}
which diverges with $n$ for any $K\geq 2n$.

This problem was the reason why the first average-case hardness proof of the \textsf{RCS}~\cite{bouland2019complexity} did not work.
The problem is solved in later studies by using a fixed-degree rational polynomial~\cite{movassagh2023hardness} or using an algorithm within a higher level of the \textsf{PH}~\cite{kondo2022quantum,bouland2022noise,krovi2022average}.
In the first approach, the author interpolated the worst-case and average-case instances of random circuits using a fixed-degree rational polynomial.
However, our problem is unlikely to admit a similar solution as we are unsure whether such an interpolation involving the permanent of a matrix even exists.
Our approach presented in the main text using Lemma~\ref{lm:bpp_poly_coefficient_bound} resembles the second approach.

\subsection{Proof of Theorem~\ref{thm:classical_hardness} for other $\mathbf{x} \in X_m$} \label{app:other_x_in_xm}

In previous subsections, we only have considered the output probability $p(\mathbf{x};J;t) = |\braket{\mathbf{x} | e^{-\complexi Ht}| \mathbf{y}_0}|^{2} $ for $\mathbf{x} = \mathbf{x}_0 = \mathbf{1}^{\sigma} \mathbf{0}^{\tau} \in X_n$.
In this case, the $2n$-th derivative of $p(\mathbf{x};J;t)$ with respect to $t$ at $t=0$ is given by the permanent of the full matrix $J$.
This subsection briefly shows that Theorem~\ref{thm:classical_hardness} still holds for other output bitstrings.
From Appendix~\ref{app:permanent_hamiltonian_moments}, we know that
\begin{align}
    \braket{\mathbf{x}|H^{l}|\mathbf{y}_0} = 0 \text{ for all } l < k, \quad \braket{\mathbf{x}|H^{m}|\mathbf{y}_0} = \frac{1}{n^m} \mathrm{Per}(J_{\mathcal{ST}}) ,
\end{align}
for $\mathbf{x} \in X_m$, where $J_{\mathcal{ST}}$ is $m \times m$ submatrix of $J$, the rows and columns of which are obtained from $\mathcal{S}=\{i \in [1,\cdots,n]: x_i = 1\}$ and $\mathcal{T}=\{i \in [1,\cdots,n]: x_{n+i} = 0\}$, respectively.
In addition, we require $m=n^{\Omega(1)}$ (e.g., $n^{0.001}$, $\sqrt{n}$, $n/2$, etc.) to make the permanent problem \textsf{\#P-hard} w.r.t. $n$.

Again, we define $p_K(t)$ as a truncated Taylor series of $p(t)=|\braket{\mathbf{x}|e^{-\complexi Ht} | \mathbf{y}_0}|^2$ (where we consider it for fixed $J$ and $\mathbf{x}$), which is given by
\begin{align}
    p_K(t) = \sum_{k=m}^K a_k t^k, \quad a_{2m} = \frac{1}{n^{2m}} \mathrm{Per}(J_{\mathcal{ST}})^2 .
\end{align}
Then, the difference between $p_K(t)$ and true polynomial $p(t)$ is given by
\begin{align}
    |p(t) - p_K(t)| \leq \frac{(2\Vert H \Vert t)^{K+1}}{(K+1)!} := \epsilon_K,
\end{align}
which is the same as before, i.e., the oracle estimates $p_K(t)$ within an additive error of $\epsilon + \epsilon_K$.
Now, we reconstruct $\hat{p}(t)=\sum_{k=2m}^{K} \hat{a}_{k} t^k$ by using a dataset obtained from the oracle that returns an estimated value of $p(t)$, and let our estimation of $\mathrm{Per}(J_{\mathcal{ST}})^{2}$ be $n^{2m} \hat{a}_{2m}$.
Collecting all the error factors, our estimation error from the recovery algorithms is written as:
\begin{align}
    \Bigl( \frac{n}{t_0} \Bigr)^{2m} (\epsilon + \epsilon_K)  \Bigl( \frac{4}{\Delta} \Bigr)^{K} {K \choose 2m}. \label{eq:permanent_additive_error_recovery}
\end{align}

Further, according to the Gaussian permanent conjecture, we know that the estimation of $\mathrm{Per}(J_{\mathcal{ST}})^{2}$ within an additive error of $\varepsilon m!$ is \textsf{\#P-hard} when $m=\mathrm{poly}(n)$.
Thus, we aim to find a condition that the error given in Eq.~\eqref{eq:permanent_additive_error_recovery} becomes smaller than $\varepsilon m!$.

Let us first consider the \textsf{FP} algorithm (Theorem~\ref{thm:fp_algorithm}).
Following the proof for $\mathbf{x}=\mathbf{x}_0$, we use $t_0 \leq C \log n$ for a constant $C$, $K=\alpha n\log n$, and $t \in [t_0(1-\Delta),t_0(1+\Delta)]$ with $\Delta < 1$.
Then, we obtain
\begin{align}
    \epsilon_K \leq \Bigl( \frac{12 e C}{\alpha} \Bigr) n^{ \alpha n \log (12 e C / \alpha)}.
\end{align}
By choosing $c_0 > \alpha \log(12eC/\alpha)$ from $\epsilon = n^{-c_0 n}$, we have $\epsilon_K \leq \epsilon$ for sufficiently large $n$, i.e., $\epsilon + \epsilon_K \leq O[n^{-\alpha n} \log(12eC/\alpha)]$.
Next, we also have
\begin{align}
    {K \choose 2m} \leq {K \choose 2n} \leq 1.5^{K} ,
\end{align}
for sufficiently large $n$.
By combining all these factors, we have
\begin{align}
    \Bigl( \frac{n}{t_0} \Bigr)^{2m} (\epsilon + \epsilon_K)\Bigl( \frac{4}{\Delta} \Bigr)^{K} {K \choose 2m} = O \biggl[ n^{2m} \Bigl( \frac{72 e C}{\alpha \Delta} \Bigr)^{\alpha n \log n} \biggr]. \label{eq:upper_bound_error_xm}
\end{align}
For any $0 \leq m \leq n$, we can find $\alpha$ such that $\alpha \log (72eC/\alpha \Delta) < -1$, which makes the bound Eq.~\eqref{eq:upper_bound_error_xm} be $o(m!)$.
This implies that $\mathrm{Per}(J_{\mathcal{ST}})^{2}$ is estimated within an additive error of $o(m!)$, which is \textsf{\#P-hard} for $m=\mathrm{poly}(n)$.

Next, we follow similar steps for the \textsf{FBPP} algorithm (Theorem~\ref{thm:classical_hardness}).
We use the same parameters for the proof of Theorem~\ref{thm:classical_hardness}, i.e., $\Delta \leq \beta / n^2$, and $K=\alpha n^3 \log n$ for a constant $\alpha,\beta >0$.
We then obtain
\begin{align}
    (\epsilon + \epsilon_K) t_0^{-2m} \Bigl( \frac{4}{\Delta} \Bigr)^K {K \choose 2m}  \leq O\Bigl[ t_0^{-2m} \Bigl( \frac{72 e t_0}{\alpha \beta} \Bigr)^{\alpha n^3 \log n} \Bigr].
\end{align}
Thus, the overall error will be
\begin{align}
    O\Bigl[ \Bigl( \frac{n}{t_0} \Bigr)^{2m} \Bigl( \frac{72 e t_0}{\alpha \beta} \Bigr)^{\alpha n^3 \log n} \Bigr] ,
\end{align}
which is $o(1)$ for any $\alpha > (72 e t_0)/\beta$.

\subsection{Remarks on Theorem~\ref{thm:classical_hardness}}
We provide the following remarks for Theorem~\ref{thm:classical_hardness}. 

\begin{enumerate}
    \item Our proof techniques for Theorem~\ref{thm:classical_hardness} can also be directly used to prove a weaker version of Conjecture~\ref{conj:gaussian_permanent}.
    In particular, one can prove that under the \textsf{BPP} reduction (i.e., such an oracle implying \textsf{FBPP=\#P}) that computing $\mathrm{Per}(J)$ or $\mathrm{Per}(J)^2$ with an additive error of $n^{-O(n)}$ with a probability of $3/4$ over random Gaussian matrices $J$, is \textsf{\#P-hard}.
    Details are provided in Appendix~\ref{app:approximate_gaussian_permanent}.

    \item Our proof for Theorem~\ref{thm:classical_hardness} is robust against the additive error for the permanent given in Conjecture~\ref{conj:gaussian_permanent}. 
    Both the \textsf{FP} (Theorem~\ref{thm:fp_algorithm}) and \textsf{FBPP} (Theorem~\ref{thm:classical_hardness}) algorithms can be feasible even for the $n^{-O(n)}$ (instead of $\varepsilon n!$) additive error estimation of the permanent by adjusting some constants.
    Thus, under the \textsf{BPP} reduction, we can directly use the results in Appendix~\ref{app:approximate_gaussian_permanent}, i.e., we do not require Conjecture~\ref{conj:gaussian_permanent}.

    \item An open question is whether the required error bounds for oracles in Theorem~\ref{thm:classical_hardness} can be relaxed.
    We believe that proving (even a restricted version of) Conjecture~\ref{conj:hardness_output_probs} by relaxing error bounds from $n^{-\alpha n^3 \log n}$ for Theorem~\ref{thm:classical_hardness} to $\varepsilon 2^{-2n}$ could be impactful.
    However, this is challenging under our current setup based on a truncated Taylor series unless significant improvements are made to the robust polynomial regression techniques.
    This is because all related approaches for \textsf{RCS} results in a similar error bound~\cite{kondo2022quantum,bouland2022noise,krovi2022average},
    and the technique we used in Lemma~\ref{lm:bpp_poly_coefficient_bound} (which was introduced in Ref.~\cite{kane2017robust}) achieves the information theoretic bound.

    \item An alternative approach for relaxing the required error bound would involve the search for a closed-form expression of $p(\mathbf{x};J;t)$ written in a combination of a finite number of elementary functions involving the permanent of $J$ or some modified matrices.
    Such an expression would facilitate extracting $\mathrm{Per}(J)^2$ more robustly, thus improving the error bound of the permanent estimations.
    We determined the possibility of writing some high-order terms from the Taylor expansion as a product of the permanent of $J$ and other complex functions of $J_{ij}$.
    Nevertheless, as other high-order terms do not allow simple decompositions, detailed analyses in this direction can be considered in future work.
\end{enumerate}

\section{Proofs of Lemmas used for Theorem~\ref{thm:classical_hardness}} \label{app:proof_details}

\subsection{Proof of Lemma~\ref{lm:p_poly_coefficient_bound}}
We prove Lemma~\ref{lm:p_poly_coefficient_bound}, which provides an \textsf{FP} algorithm for estimating the derivatives of a polynomial function when noisy data are provided.

\begin{theorem}[Restatement of Lemma~\ref{lm:p_poly_coefficient_bound}] \label{thm:p_poly_coefficient_bound_detail}
    Let us assume that $q(t)$ is a degree-$d$ polynomial and $0 < \Delta < 1$.
    Suppose that a dataset $\{(t_i, y_i)\}_{i=1}^L$ such that $\{t_i\}$ are $L=d+1$ equidistant points in $[t_0(1-\Delta), t_0(1+\Delta)]$ with $\Delta < 1$ and $|y_i-q(t)| \leq \delta$ for all $i$ is given.
    Then the $k$-th coefficient of $q(t)$, i.e., $q^{(k)}(0)/k!$, can be estimated in $\mathrm{poly}(d)$ time within an error $O[\delta t_0^{-k} (4/\Delta)^{d} {d \choose k}]$.
\end{theorem}
\begin{proof}
    We construct a degree-$d$ polynomial, $\hat{q}(t)$, by using the Lagrange interpolation from the given dataset and select its coefficient for $t^k$ as our estimation of the $k$-th coefficient of $q(t)$, i.e., $q^{(k)}(0)/k!$.
    Thus, we aim to bound the error between the $k$-th coefficient of two polynomials $\hat{q}(t)$ and ${q}(t)$.
    By writing $r(t) := q(t) - \hat{q}(t) = \sum_{k=0}^{d} b_k t^k$, the error we want to bound is $b_k$.

    Let us define
    \begin{align}
        R(s) &:=r ( t=t_0(\Delta s +1) \bigr) = \sum_{k=0}^{d} c_k s^k.
    \end{align}
    Then, condition $|q(t_i) - y_i| \leq \delta$ for $L$ equidistant points $t_i \in [t_0(1-\Delta), t_0(1+\Delta)]$ implies that $|R(s)| \leq \delta$ for $L$ equidistant points between $s \in [-1, 1]$. 
    By using Rakhmanov's theorem~\cite{rakhmanov2007bounds}, we obtain
    \begin{align}
        |R(s)| \leq C \delta ,
    \end{align}
    for all $|s| < 1$ where $C >0$ is a constant.
    According to the lemma below, we have
    \begin{align}
        \sum_{k} |c_k| \leq C 4^{d} \delta.
    \end{align}
    We now obtain a relation between $b_k$ and $c_k$ by inserting $s=( t/t_0 - 1)/\Delta$ to $R(t)$, which yields
    \begin{align}
        b_k = \frac{1}{t_0^k} \sum_{l \geq k}^{d} (-1)^{(l-k)} \frac{c_l}{\Delta^l} {l \choose k}.
    \end{align}
    Therefore,
    \begin{align}
        |b_k | &\leq \frac{1}{t_0^k} \Bigl( \sum_{l \geq k}^{d} |c_l| \Bigr) \Delta^{-d} {d \choose k} \leq C \delta \Bigl(\frac{4}{\Delta}\Bigr)^{d} t_0^{-k} {d \choose k},
    \end{align}
    which completes the proof.
\end{proof}

\begin{remark}
The proof can be simplified without invoking Rakhmanov's theorem but by directly using the setting in Ref.~\cite{sherstov2012making}.
However, this does not change the overall scaling behavior.
\end{remark}

\begin{lemma}[Lemma 4.1 in Ref.~\cite{sherstov2012making}] \label{lm:sum_of_coeffs}
    Let $Q(x) = \sum_{i=0}^d a_i x^i$ be a polynomial. Then
    \begin{align}
        \sum_{k=0}^d |a_k| \leq 4^d \max_{x \in [-1,1]} |Q(x)| .
    \end{align}
\end{lemma}

One subtlety of the algorithm is that it returns $\hat{q}(t)$ in a representation involving Lagrange basis polynomials, i.e.,
we have $\hat{q}(t) = \sum_{i=1}^{K+1} \alpha_i l_i(t)$ where
\begin{align}
    l_i(t) = \prod_{j \neq i} \frac{t - t_j}{t_i - t_j}.
\end{align}
Thus, to extract the $k$-th coefficient of $\hat{q}(t)$, the $k$-th coefficient of $l_i(t)$ must be efficiently obtained.
The following Lemma shows that all the coefficients can be obtained in polynomial time.
\begin{lemma}
    For a given degree-$d$ polynomial, which is represented using its roots $\{r_k\}$, i.e., $q(t) = \prod_{k=1}^d (t - r_k)$, one can obtain all the coefficients in $O(d \log^2 d)$ time.
\end{lemma}
\begin{proof}
    Let us define $q_1(t) = \prod_{k=1}^{d/2} (t - r_k)$ and $q_2(t) = \prod_{k=d/2+1}^{d} (t - r_k)$.
    When coefficients of $q_1(t)$ and $q_2(t)$ are known, coefficients of $q(t)=q_1(t) q_2(t)$ can be obtained using the fast Fourier transformation in $O[d \log (d)]$ time. Then, the recurrence relation is given as follows:
    \begin{align}
        T(d) = 2T(d/2) + O(d \log(d)),
    \end{align}
    where $T(d)$ is the time complexity for a degree-$d$ polynomial. Solving the equation gives $T(d) = O[d \log^2(d)]$.
\end{proof}

As mentioned earlier, one of the limitations of Theorem~\ref{thm:p_poly_coefficient_bound_detail} is that the algorithm requires a dataset with $L=d+1$ equidistant points, which is only possible when the oracle correctly outputs with a probability greater than $1-1/L$.
We also note that reducing the required number of samples is possible using a dataset optimally chosen for Chebyshev polynomials as in Ref.~\cite{kondo2022quantum}.
However, such a technique does not improve the overall scaling of the error.

\subsection{Proof of Lemma~\ref{lm:bpp_poly_coefficient_bound}}

\begin{theorem}[Restatement of Lemma~\ref{lm:bpp_poly_coefficient_bound}]
    Let us assume that $q(t)$ is a degree-$d$ polynomial and $0 < \Delta < 1$.
    Suppose that we can obtain samples of $q(t)$ within error $\delta$ with a constant probability greater than $1/2$ over uniformly random $t \in [t_0(1-\Delta), t_0(1+\Delta)]$, i.e.,
    \begin{align}
        \mathrm{Pr}[|y_i-q(t_i)| \leq \delta] > 1/2, \text{ where } t_i \sim \mathcal{U}_{[t_0(1-\Delta), t_0(1+\Delta)]}.
    \end{align}
    Then, there exists an \emph{\textsf{FBPP}} algorithm that returns the $k$-th coefficient of $q(t)$, i.e., $q^{(k)}(0)/k!$, within an additive error of $O[ \delta t_0^{-k} (4/\Delta)^d {d \choose k}] $ with a success probability of at least $2/3$.
\end{theorem}

We prove the theorem by using the following lemma.

\begin{lemma}[Ref.~\cite{kane2017robust}]\label{lm:bpp_poly_recon}
    Let $Q(s)$ be a degree-$d$ polynomial defined in $s \in [-1,1]$ and suppose that we can obtain samples of $Q(s)$ within error $\delta$ with probability $\rho > 1/2$ where each $s_i$ is drawn from the Chebyshev distribution $D_c(s) \sim \frac{1}{\sqrt{1-s^2}}$.
    Then, by using $O[\frac{d}{\epsilon}\log \frac{d}{\epsilon \chi}]$ number of samples, we can find a polynomial $\hat{Q}(s)$ that satisfies
    \begin{align}
        \max_{t \in [-1, 1]} |Q(s) - \hat{Q}(s)| \leq (2+\epsilon) \delta ,
    \end{align}
    with probability $1-\chi$, where $\epsilon \in (0,1/2)$ is a constant.

    On the other hand, if $s_i$ is drawn from the uniform distribution instead, then $O(\frac{d^2}{\epsilon^2} \log \chi^{-1})$ samples suffice for the same result.
    In both cases, the runtime of this algorithm is polynomial in the sample size.
\end{lemma}

We prove the theorem by applying this lemma to $Q(s) := q(t_0(1+s\Delta))$.
By putting $\epsilon=1/4$ into Lemma~\ref{lm:bpp_poly_recon}, we can find $\hat{Q}(s)$ such that
\begin{align}
    \max_{x\in[-1,1]} |Q(s) - \hat{Q}(s)| \leq \frac{9}{4} \delta,
\end{align}
with a probability greater than $2/3$ by using $O(d^2)$ samples.
Then, we take the $k$-th coefficient of $\hat{q}(t):=\hat{Q}((t/t_0-1)/\Delta)$ as our estimation of $q^{(k)}(0)/k!$.

By writing $R(s):=Q(s) - \hat{Q}(s) = \sum_{k=0}^d c_k s^k$, we obtain
\begin{align}
    r(t) := q(t) - \hat{q}(t) = R\Bigl( \frac{t/t_0 - 1}{\Delta} \Bigr) = \sum_{k=0}^d b_k t^k .
\end{align}
Again, from
\begin{align}
    b_k = \frac{1}{t_0^k} \sum_{l \geq k}^{d} (-1)^{(l-k)} \frac{c_l}{\Delta^l} {l \choose k},
\end{align}
and Lemma~\ref{lm:sum_of_coeffs},
we know that
\begin{align}
    |b_k| \leq \frac{9 \delta}{4 t_0^k} \Bigl(\frac{4}{\Delta}\Bigr)^{d} {d \choose k},
\end{align}
which is the error from our estimation.

\subsection{Probability to have $\Vert H \Vert \leq 3n$}

By taking a term-wise norm, the infinite norm of $H_{1,2,3,4}$ is bounded by
\begin{align}
    \Vert H_{1,3} \Vert \leq \sum_{ij} |J_{ij}|/n , \quad \Vert H_{2} \Vert \leq 2 \sum_{ij} |J_{ij}|/n , \quad \Vert H_{4} \Vert \leq \frac{3}{2} \sum_{ij} |J_{ij}|/n.
\end{align}
Hence, the question is when $|J_{ij}| \leq 3 n^2$ (for $H_{1,3}$) , $3n^2/2$ (for $H_2$), and $2n^2$ (for $H_4$) are satisfied. As each $J_{ij}$ is a random variable, we can use the lemma below to have
\begin{align}
    \mathrm{Pr}[\sum_{ij} |J_{ij}| \leq cn^2] \geq 1 - 2^{n^2} e^{-c^2 n^2 /2}.
\end{align}

By choosing $c=3$ for $H_{1,3}$, $c=3/2$ for $H_2$, and $c=2$ for $H_4$, we obtain
\begin{align}
    \mathrm{Pr}[ \Vert H_{1,3} \Vert \leq 3n] \geq 1 - \Bigl( \frac{2}{e^9} \Bigr)^{n^2} ,\\
    \mathrm{Pr}[ \Vert H_{2} \Vert \leq 3n] \geq 1 - \Bigl( \frac{2}{e^{9/4}} \Bigr)^{n^2} ,\\
    \mathrm{Pr}[ \Vert H_{4} \Vert \leq 3n] \geq 1 - \Bigl( \frac{2}{e^4} \Bigr)^{n^2}.
\end{align}

\begin{lemma} \label{lm:bounding_prob_sum_of_gaussian}
    Let $\{X_i\}$ be $n$ independent random Gaussian variables, i.e., $X_i \sim \mathcal{N}(0,1)$ for all $i \in \{1, \cdots, n\}$. Then the following bound is satisfied:
    \begin{align}
        \mathrm{Pr}\Bigl[ \frac{1}{n} \sum_{i=1}^n |X_i| \geq t \Bigr] \leq 2^n e^{-nt^2/2}.
    \end{align}
\end{lemma}
\begin{proof}
    Let $X\sim \mathcal{N}(0,1)$. Then the standard Gaussian integration gives the moment generating function $\mathbb{E}[e^{tX}] = e^{t^2/2}$ for $t \in \mathbb{R}$. In addition, 
\begin{align}
    \mathbb{E}[e^{t|x|}] \leq \mathbb{E}[e^{t|x|} + e^{-t|x|}] = 2e^{t^2 / 2}.
\end{align}

We now consider $n$ Gaussian random variables . For any $\lambda > 0$, we have
\begin{align*}
    \mathrm{Pr}\Bigl[ \frac{1}{n} \sum_{i=1}^n |X_i| \geq t \Bigr] &= \mathrm{Pr} \Biggl[ \exp\Bigl( \frac{\lambda}{n} \sum_{i=1}^n |X_i| \Bigr)  \geq e^{\lambda t} \Biggr] \\
    &\leq e^{-\lambda t} \mathbb{E} \Bigg[ \exp\Bigl( \frac{\lambda}{n} \sum_{i=1}^n |X_i| \Bigr)\Bigg] = e^{-\lambda t}  \prod_{i=1}^n \mathbb{E}\Bigg[ \exp\Bigl( \frac{\lambda}{n} |X_i| \Bigr)\Bigg] \\
    &\leq e^{-\lambda t} 2 e^{\lambda^2/(2n^2)} = 2^n e^{\lambda^2/(2n)-\lambda t}, \numberthis
\end{align*}
where we used the standard Chernoff bound for the first inequality.

By taking $\lambda = nt$, which minimizes the argument of the exponential, we obtain the desired inequality.
\end{proof}

\section{Average-case approximate hardness of a Gaussian matrix permanent} \label{app:approximate_gaussian_permanent}

The arXiv version of Ref.~\cite{Aaronson2011} proved that estimating the permanent of Gaussian matrices within an exponentially small additive error is difficult (Theorem~62 therein), which supports Conjecture~\ref{conj:gaussian_permanent}.
This appendix provides an alternative proof giving an exact required error bound for estimation.

We consider the following well-known fact that computing permanent of 0/1 matrices is \textsf{\#P-hard} for worst-case instances.
\begin{theorem}[Ref.~\cite{valiant1979complexity}]
    There exists a matrix $X \in \mathbb{R}^{n \times n}$, the elements of which are all $0$ or $1$ such that computing $\mathrm{Per}(X)$ is \emph{\textsf{\#P-hard}}.
\end{theorem}

Note that this readily implies that computing the permanent of 0/1 matrix within an additive error of $o(1)$ is also \textsf{\#P-hard} given that the result must be an integer. 
We now prove that estimating the permanent of a Gaussian matrix is also hard on average.

\begin{theorem}\label{thm:approximate_permanent_gaussian}
    If there exists an efficient classical algorithm that estimates $\mathrm{Per}(J)$ within an additive error of $n^{-2n(1+\delta)}$ for any $\delta>0$ with probability $3/4$ over $J \sim \mathcal{N}(0,1)^{n \times n}$, then \emph{\textsf{FBPP=\#P}}.
\end{theorem}

Following Aaronson and Arkhipov~\cite{Aaronson2011}, we choose a matrix $X(t):=tX+(1-t)Y$ where $X$ is a 0/1 matrix for the worst-case instance, and $Y$ is a matrix drawn from a random Gaussian distribution.
Based on the definition of the matrix permanent, $q(t) = \mathrm{Per}(X(t))$ is a degree-$n$ polynomial.
Under this setting, the following lemmas will be used to prove the theorem.

\begin{lemma} \label{lm:tvd_gaussian_permanent}
    The total variational distance between $X(t):=tX + (1-t)Y$ for a random Gaussian matrix $Y$ and $X=(x_{ij})$ is bounded above by $n(3 \sqrt{t} + t)$.
\end{lemma}
\begin{proof}
    For a given $t$, we can consider $X(t) \sim \prod_{ij} \mathcal{N}(t x_{ij}, (1-t)^2):=D_t$.
    Based on results in Ref.~\cite{devroye2018total}, for $D_{t}' := \prod_{ij} \mathcal{N}(t x_{ij}, 1^2)$, we have
    \begin{align}
        \Vert D_t - D_t' \Vert \leq n \frac{3}{2} \sqrt{1-(1-t)^2} < 3n \sqrt{t}.
    \end{align}
    Note that Ref.~\cite{Aaronson2011} used a different bound given by $\Vert D_t - D_t' \Vert \leq n^2 t$ (Lemma 48, therein).
    However, as we could not follow that proof, we here use a weaker bound proved in Ref.~\cite{devroye2018total}, which does not change the overall scaling.
    In addition, we have
    \begin{align}
        \Vert D_t' - \mathcal{N}(0,1)^{n \times n} \Vert \leq \sqrt{\sum_{ij} (t x_{ij})^2} \leq nt ,
    \end{align}
    as shown in Ref.~\cite{Aaronson2011}.
    Then,
    \begin{align}
        \Vert D_t - \mathcal{N}(0,1)^{n \times n} \Vert &\leq \Vert D_t - D_t' \Vert + \Vert D_t' - \mathcal{N}(0,1)^{n \times n} \Vert \nonumber \\
        &\leq n(3\sqrt{t} + t).
    \end{align}
\end{proof}

\begin{lemma}[Ref.~\cite{krovi2022average}] \label{lm:error_noisy_polynomial_at1}
    Let us assume that $q(t)$ is a degree-$d$ polynomial and $\Delta \in (0,1)$.
    Suppose that we can obtain samples of $q(t)$ within an additive error of $\delta$ with a constant probability of $\eta > 1/2$ for $t_i$ drawn uniformly from $[-\Delta, \Delta]$.
    Then, there exists an \emph{\textsf{FBPP}} algorithm that returns $q(1)$ within an additive error of $\frac{9\delta}{4} (4/\Delta)^d$.
\end{lemma}

\begin{proof}
    Let us define $Q(s) := q(\Delta s)$, a degree-$d$ polynomial in $s$.
    Under the condition used in the lemma, an oracle can estimate $Q(s)$ for $s \in [-1,1]$ within an additive error of $\delta$.
    By using such an oracle, we can construct a polynomial $\hat{Q}(s)$ following Lemma~\ref{lm:bpp_poly_recon}, which satisfies
    \begin{align}
        \max_{s \in [-1,1]} |Q(s) - \hat{Q}(s)| \leq \frac{9}{4}\delta.
    \end{align}
    Then for $\hat{q}(t) = \hat{Q}(t/\Delta)$, our estimation of $q(1)$ will be $\hat{q}(1)$.
    
    We now bound the error between $q(1)$ and $\hat{q}(1)$.
    Let us define $R(s) := Q(s) - \hat{Q}(s) = \sum_{k=0}^d c_k s^k$.
    By Lemma~\ref{lm:sum_of_coeffs}, we obtain
    \begin{align}
         \sum_{k=0}^d |c_k| \leq \frac{9\delta}{4} 4^d .
    \end{align}
    Then, it follows that
    \begin{align}
        |q(1) - \hat{q}(1)| = |R(\Delta^{-1})| \leq \sum_{k=0}^d |c_k| \Delta^{-k} \leq \frac{9 \delta}{4} \Bigl( \frac{4}{\Delta} \Bigr)^d.
    \end{align}
\end{proof}

\textbf{Proof of Theorem~\ref{thm:approximate_permanent_gaussian}}
Let us assume an oracle that estimates $\mathrm{Per}(J)$ within an additive error of $\epsilon$, with a constant probability of $3/4$ over a random Gaussian matrix $J$.
Then from Lemma~\ref{lm:tvd_gaussian_permanent}, such an oracle also estimates $q(t)=\mathrm{Per}(X(t))$ with probability $\eta = 3/4 - n (3\sqrt{\Delta} + \Delta)$ when $t \in [-\Delta, \Delta]$.
By taking $\Delta = (16n)^{-2}$, we could obtain $\eta > 1/2$ for all $n$.

Next, we apply Lemma~\ref{lm:error_noisy_polynomial_at1} to $q(t)=\mathrm{Per}(X(t))$. As the degree of $q(t)$ is $n$, and the error is $\delta=\epsilon$, the error for $q(1)=\mathrm{Per}(J)$ becomes 
\begin{align}
    \frac{9\epsilon}{4} \Bigl( \frac{4}{(16n)^{-2}} \Bigr)^n = \frac{9\epsilon}{4} (2^{10} n^{2})^n.
\end{align}
The error is upper bounded by $o(1)$ when $\epsilon = n^{-2n(1+\delta)}$ for any $\delta > 0$.
This implies that estimating $q(t)$ at $t=0$ within an additive error of $\epsilon$ is \textsf{\#P-hard} under the \textsf{BPP} reduction.

The same argument as above also applies to $\mathrm{Per}(J)^2$ instead of $\mathrm{Per}(J)$.
From Theorem~28 in Ref.~\cite{Aaronson2011}, we know that the worst-case instance $X \in \mathbb{R}^{n \times n}$ exists such that computing $\mathrm{Per}(X)^2$ is \textsf{\#P-hard}.
Then, we consider a polynomial $X(t):= tX + (1-t)Y$ where $Y$ is a random matrix from $\mathcal{N}(0, 1)^{n \times n}$.
As $\mathrm{Per}(X(t))^2$ is a degree-$2n$ polynomial, we can use the same steps to obtain the following theorem:
\begin{theorem}
    If there exists an efficient classical algorithm that estimates $\mathrm{Per}(J)^2$ within an additive error of $n^{-4n(1+\delta)}$ for any $\delta>0$ with probability $3/4$ over $J \sim \mathcal{N}(0,1)^{n \times n}$, then \emph{\textsf{FBPP=\#P}}.    
\end{theorem}

Compared to Theorem~62 in Ref.~\cite{Aaronson2011}, 
the required probability of our theorem is $3/4$ (in fact, any constant probability larger than $1/2$ can be used) instead of $3/4 + 1/\mathrm{poly}(n)$ and the required additive error $n^{-4n(1+\delta)}$ does not depend on the probability.

\section{Proof details of classical hardness of the sampling task} \label{app:hardness_sampling_proof}
In Appendix~\ref{app:complexity_output_dist_proof}, we proved that computing the output distribution is a classically hard problem because it belongs to class \textsf{\#P-hard}.
In this section, we further prove that \textit{sampling} from the output distribution is also a classically challenging problem unless the \textsf{PH} collapses.

We follow the usual arguments in Refs.~\cite{Aaronson2011,bremner2016average,boixo2018characterizing}, which combine Stockmeyer's theorem and the anticoncentration property of the output distribution to argue that sampling from the distribution is also classically hard (unless the \textsf{PH} collapses).
In these arguments, Stockmeyer's theorem provides $\mathsf{FBPP}^\mathsf{NP}$ algorithm for estimating the output probability within an inverse polynomial multiplicative error, 
which is further reduced to an additive error utilizing the anticoncentration property.

However, as we mentioned in the main text, these techniques cannot be directly applied to our case, as our complexity results (Theorem~\ref{thm:classical_hardness}) are only satisfied for $\mathbf{x} \in X_{m}$ with sufficiently large $m$.
This is related to the fact that our problem lacks the hiding property.
The main result obtained in this section is that such restricted complexity results are sufficient to prove the difficulty of sampling when the error between the distribution we sample from and the true distribution $p(\mathbf{x};J;t)$ is inverse polynomially small, provided that the output distribution anticoncentrates at least for a constant fraction of $\mathbf{x} \in X_{n/2}$.

Let us start with Stockmeyer's approximate counting theorem.
\begin{theorem}[Stockmeyer~\cite{Stockmeyer1985}] \label{thm:stockmeyer}
    Suppose that there exists an oracle $\mathcal{O}$ that generates a sample from the probability distribution $p(\mathbf{x})$ by taking a random bitstring $\mathbf{r}$, i.e.,
    \begin{align}
        p(\mathbf{x}) = \underset{\mathbf{r} \in \{0,1\}^{\mathrm{poly}(n)} }{\mathrm{Pr}} \bigl[ \mathcal{O}(\mathbf{r}) = \mathbf{x} \bigr].
    \end{align}
    There exists an $\mathsf{FBPP}^{\mathsf{NP}^{\mathcal{O}}}$ machine that, for a given $\mathbf{x} \in \{0,1\}^n$, approximates $p(\mathbf{x})$ within a multiplicative factor of $g = 1/\mathrm{poly}(n)$.
    In other words, for a given $\mathbf{x}$, the machine outputs $\hat{p}(\mathbf{x})$ such that
    \begin{align}
        (1-g)\hat{p}(\mathbf{x}) \leq p(\mathbf{x}) \leq (1+g)\hat{p}(\mathbf{x}).
    \end{align}
\end{theorem}

Then, the following theorem provides an argument that a classical approximate sampler can estimate $p(\mathbf{x})$ for $\mathbf{x} \in X_{n/2}$.

We now prove Theorem~\ref{thm:apply_stockmeyer} introduced in the main text.
\applystockmeyer*

\begin{proof}
    We first prove the case for $H_{1,2}$. From the definition of $L_1$ distance
    \begin{align}
        \nu = \sum_{\mathbf{x} \in \mathbb{Z}_2^{2n}} |p(\mathbf{x}) - p'(\mathbf{x})|,
    \end{align}
    it follows from Markov's inequality that
    \begin{align}
        \underset{\mathbf{x} \in \mathbb{Z}_2^{2n}}{\mathrm{Pr}} \Bigl[ |p(\mathbf{x}) - p'(\mathbf{x})| \geq \frac{\nu}{\delta} 2^{-2n} \Bigr] \leq \delta.
    \end{align}
    Next, the probability to have $\mathbf{x} \in X_{n/2}$ among all $\mathbf{x} \in \mathbb{Z}_2^{2n}$ is given by
    \begin{align}
        \frac{{n \choose n/2}^2}{2^{2n}} \geq \frac{1}{2n} ,
    \end{align}
    where we have used the inequality
    \begin{align}
        {n \choose k} \geq \sqrt{\frac{n}{8 k(n-k)}}2^{n \mathcal{H}(k/n)} .
    \end{align}
    Then, by taking $\delta = \gamma (2n)^{-1}$, more than $1-\gamma$ fraction of $\mathbf{x} \in X_{n/2}$ satisfies
    \begin{align}
        |p(\mathbf{x}) - p'(\mathbf{x})| \leq \gamma^{-1} \nu (2n) 2^{-2n}. \label{eq:prob_dist_bound_markov_h12}
    \end{align}

    Theorem~\ref{thm:stockmeyer} implies that if there exists an efficient classical sampler for $p'(\mathbf{x})$, then there exists an $\mathsf{FBPP}^{\mathsf{NP}^{\mathcal{O}}}$ algorithm that outputs $\hat{p}(\mathbf{x})$ such that
    \begin{align}
        |\hat{p}(\mathbf{x}) - p'(\mathbf{x})| \leq g p'(\mathbf{x}) ,
    \end{align}
    where $g = 1/\mathrm{poly}(n)$.
    Then the distance between $\hat{p}(\mathbf{x})$ and $p(\mathbf{x})$ for $\mathbf{x}$ satisfying Eq.~\eqref{eq:prob_dist_bound_markov_h12} is given by
    \begin{align*}
        |\hat{p}(\mathbf{x}) - p(\mathbf{x})| &\leq |\hat{p}(\mathbf{x}) - p'(\mathbf{x})| + |p(\mathbf{x}) - p'(\mathbf{x})| \\
        &\leq g p'(\mathbf{x} ) + |p(\mathbf{x}) - p'(\mathbf{x})| \\
        &\leq (1+g) |p(\mathbf{x} ) - p'(\mathbf{x} )| + gp(\mathbf{x} ) \\
        &\leq (1+g) \gamma^{-1} \nu (2n) 2^{-2n} + g p(\mathbf{x}). \numberthis
    \end{align*}

    Proof for $H_{3,4}$ follows the same step as above, besides that we only consider $\mathbf{x} \in X$ instead of $\mathbb{Z}_2^{2n}$.
    Recall that $X= \cup_{i=0}^n X_i$ is the set of all bitstrings that spans the eigenvalue zero subspace of $J_z$.
    From the definition of $L_1$ distance
    \begin{align}
        \nu = \sum_{\mathbf{x} \in X} |p(\mathbf{x}) - p'(\mathbf{x})|,
    \end{align}
    and Markov's inequality, we obtain
    \begin{align}
        \underset{\mathbf{x} \in X}{\mathrm{Pr}} \Bigl[ |p(\mathbf{x}) - p'(\mathbf{x})| \geq \frac{\nu}{\delta} {2n \choose n}^{-1} \Bigr] \leq \delta.
    \end{align}
    Here, ${2n \choose n} =|X|$.
    In addition, the probability to have $\mathbf{x} \in X_{n/2}$ among all allowed $x$ (within the same $J_z$ eigenspace) is given by
    \begin{align}
        \frac{{n \choose n/2}^2}{{2n \choose n}} \geq \frac{\sqrt{\pi}}{2} n^{-1/2},
    \end{align}
    where we have used the inequality
    \begin{align}
        \sqrt{\frac{n}{8 k(n-k)}}2^{n \mathcal{H}(k/n)} \leq {n \choose k} \leq \sqrt{\frac{n}{2 \pi k(n-k)}}2^{n \mathcal{H}(k/n)}.
    \end{align}
    Thus by taking $\delta = \gamma (\sqrt{\pi}/2) n^{-1/2}$, more than $1-\gamma$ fraction of $\mathbf{x} \in X_{n/2}$ satisfies
    \begin{align}
        |p(\mathbf{x}) - p'(\mathbf{x})| \leq (\sqrt{\pi} \gamma/2)^{-1} \nu n^{1/2} {2n \choose n}^{-1}. \label{eq:prob_dist_bound_markov}
    \end{align}
    Using the same steps for $H_{1,2}$, the desired error $\epsilon_{II}$ is obtained.
\end{proof}

One can see that the error $\epsilon_I$ and $\epsilon_{II}$ are smaller than that stated in Conjecture~\ref{conj:hardness_output_probs} denoted by $\varepsilon 2^{-2n}$ and $\varepsilon {2n \choose n}^{-1}$, respectively, when $p(\mathbf{x})$ also scales with the inverse of the dimension, i.e., $2^{-2n}$ and ${2n \choose n}^{-1}$.
A version of this condition, so-called the anticoncentration property~\cite{Aaronson2011,bremner2016average}, is widely known for the sampling problems and can be rigorously proven for the \textsf{IQP}~\cite{bremner2016average,bermejo2018architectures} and the \textsf{RCS}~\cite{brandao2016local}.
We also expect that this property is satisfied for Hamiltonian dynamics, at least for some $\mathbf{x} \in X_{n/2}$, which we conjecture as follows:

\anticoncentration*

When the anticoncentration conjecture holds, we can apply Theorem~\ref{thm:apply_stockmeyer} by choosing $\gamma = \zeta/2$, which makes at least $\zeta/2$ fraction of $\mathbf{x} \in X_{n/2}$ must have both properties. 
Then for such $\mathbf{x}$, by taking $\nu = \varepsilon \gamma /(4n)$ for $H_{1,2}$ (or $\nu = \varepsilon (\sqrt{\pi}\gamma/4) n^{-1/2} $ for $H_{3,4}$),
one can estimate $p(\mathbf{x};J)$ within an additive error of $[\varepsilon(1+g)/2 + g \alpha] D^{-1}$ with a probability greater than $1-\beta$ over $J$.
Here, $D$ is the dimension of the Hilbert space given by $2^{2n}$ and ${2n \choose n}$ for $H_{1,2}$ and $H_{3,4}$, respectively.
In addition, as $g=1/\mathrm{poly}(n)$, $\varepsilon(1+g)/2 + g \alpha \leq \varepsilon$ for sufficiently large $n$.

In summary, we obtained the following result presented in the main text.
\ComplexitySmallTVDSampling*

Note that Corollary~\ref{col:complexity_small_tvd_sampling} is robust in the sense that one can change the required fraction of $\mathbf{x}$ showing the anticoncentration property by adjusting the $L_1$ distance $\nu$.
For example, the proof for $H_{3,4}$ is feasible with $\nu = O(\varepsilon)$ when a constant fraction of $\mathbf{x} \in X = \cup_{i=0}^n X_{i}$ satisfies the anticoncentration condition.
On the other hand, if only $\xi=O(1/n)$ fraction of $\mathbf{x}\in X_{n/2}$ shows the anticoncentration property, the required $L_1$ distance becomes $\nu = O(\varepsilon n^{-2})$  and $\nu = O(\varepsilon n^{-3/2})$ for $H_{1,2}$ and  $H_{3,4}$, respectively.

\begin{figure*}
    \centering
    \includegraphics[width=0.9\linewidth]{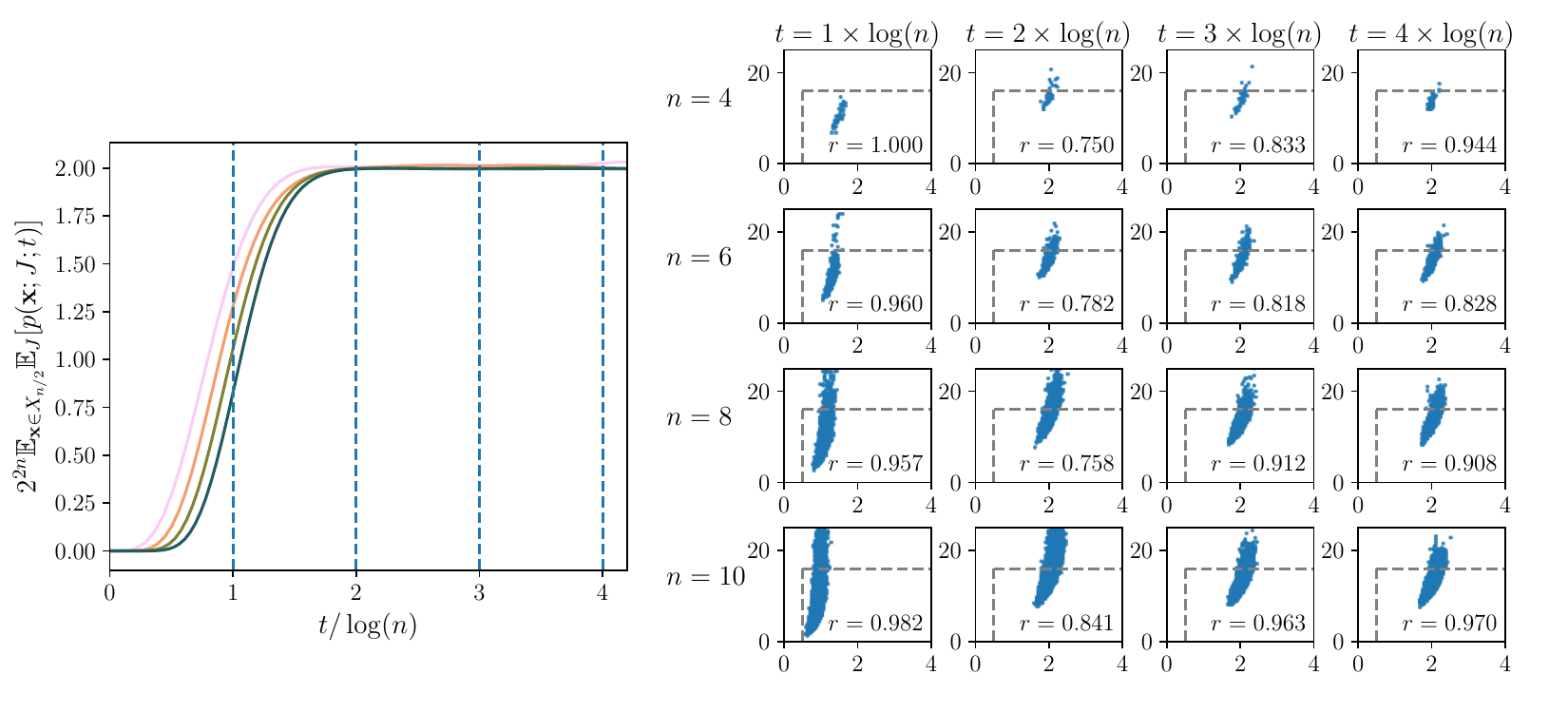}
    \caption{Left: Output probability $p(\mathbf{x};J;t)$ averaged over $J\sim \mathcal{N}(0,1)^{n \times n}$ and $\mathbf{x} \in X_{n/2}$ as a function of $t$.
    We plot results for $n=4$ (the lightest) to $n=10$ (the darkest) with step size $2$. The plot indicates that time to equilibrate $t_{\rm eq}$ scales with $\log n$. 
    Right: Scatter plots of $2^{2n} \mathbb{E}_J[p(\mathbf{x};J;t)]$ (x-axis) versus $2^{4n} \mathbb{E}_J[p(\mathbf{x};J;t)^2]$ (y-axis) for each $\mathbf{x} \in X_{n/2}$.
    For $t/\log n \in [1,2,3,4]$ and $n \in [4,6,8,10]$, we numerically compute $p(\mathbf{x};J;t)$ and $p(\mathbf{x};J;t)^2$  for all $\mathbf{x} \in X_{n/2}$ averaged over $2^{10}$ random instances of $J\sim \mathcal{N}(0,1)^{n \times n}$.
    The dashed lines indicate $2^{2n} \mathbb{E}_J[p(\mathbf{x})] = 1/2$ and $ 2^{4n} \mathbb{E}_J[p(\mathbf{x})^2] = 16$.
    The plot also presents $r$, which is the ratio of $\mathbf{x} \in X_{n/2}$ satisfying $2^{2n} \mathbb{E}_J[p(\mathbf{x})] \geq 1/2$ and $2^{4n} \mathbb{E}_J[p(\mathbf{x})^2] \leq 16$.
    For numerical simulations, we used the expansion $e^{-\complexi H_1 t} = \prod_{i,j=1}^n [e^{-i (J_{ij}/n)\sigma^{(i)}_x \tau^{(j)}_x t}]$.
    }
    \label{fig:anticon_h1}
\end{figure*}

\section{Anticoncentration of the Ising model}~\label{app:anticon_ising}
In the main text, we mainly discussed the anticoncentration of the output distribution for $H_3$.
In this appendix, we study the anticoncentration property for the Ising model ($H_1$).

\subsection{Numerical demonstration of anticoncentration}
As discussed in the main text, we demonstrated anticoncentration of the output distribution for $H_3$ by computing $\mathbb{E}_J[p(\mathbf{x};J;t)]$ and $\mathbb{E}_J[p(\mathbf{x};J;t)^2]$
where $p(\mathbf{x};J;t) = |\braket{\mathbf{x}| e^{-\complexi H_1 t} | \mathbf{y}_0 }|^2$.
Similarly, we also compute the same quantity for $H=H_1$, as plotted in Fig.~\ref{fig:anticon_h1}.

In Fig.~\ref{fig:anticon_h1}(Left), we observe that $\mathbb{E}_{\mathbf{x} \in X_{n/2}}\mathbb{E}_J[p(\mathbf{x};J;t)]$ converges to $\approx 2.0 \times 2^{-2n}$ for large $t$ regardless of $n$.
We analytically prove in the next subsection that this result is related to the Hamming weight of $\mathbf{x}$, which is $n$ for $\mathbf{x} \in X$.
In addition, the plot shows that the time for equilibration scales with $\log n$.

Furthermore, we plot $\mathbb{E}_J[p(\mathbf{x};J;t)]$ and $\mathbb{E}_J[p(\mathbf{x};J;t)^2]$ for each $\mathbf{x} \in X_{n/2}$ in Fig.~\ref{fig:anticon_h1}(Right).
We also present $r$ which is the ratio of $\mathbf{x}$ satisfying
\begin{align}
    \mathbb{E}_J[p(\mathbf{x};J;t)] \geq \frac{1}{2} 2^{-2n} \land
    \mathbb{E}_J[p(\mathbf{x};J;t)^2] \leq 16 \times 2^{-4n}. \label{eq:anticon_ising_eq}
\end{align}
The results show that $r$ converges to a non-zero value as $n$ increases for all $t/\log n \in [1,2,3,4]$.
Then it follows from the Paley-Zygmund inequality that the output probability anticoncentrates.

\subsection{Converged output probability of the Ising model}
In this subsection, we analytically derive the behavior of the output probability $p(\mathbf{x};J;t)$ when $t\rightarrow \infty$.
Our analytic expression shows that the converged value of the output probability depends on the Hamming weight of $\mathbf{x}$.
For convenience, we divide $\mathbf{x}$ into $\x^\sigma = \{x_1,\cdots,x_n\}$ and $\x^\tau = \{x_{n+1},\cdots,x_{2n}\}$.

As the first step, we expand $f(\mathbf{x};J;t):=\braket{\mathbf{x}^{\sigma} \mathbf{x}^{\tau} | e^{-iHt} | \mathbf{y}_0}$, which is given by
\begin{align*}
    f(\mathbf{x};J;t) &= \braket{\mathbf{x}^{\sigma} \mathbf{x}^{\tau} | e^{-iHt} | \mathbf{y}_0} = \braket{\mathbf{x}^{\sigma} \mathbf{x}^{\tau} | e^{-i \sum_{ij}J_{ij} \sigma^{(i)}_x \tau^{(j)}_x t} | \mathbf{y}_0} \\
    &= \braket{\mathbf{x}^{\sigma} \mathbf{x}^{\tau} | H^{\otimes 2n} e^{-i \sum_{ij} J_{ij}\sigma^{(i)}_z \tau^{(j)}_z t/n} H^{\otimes 2n} | \mathbf{y}_0} \\
    &= \frac{1}{\sqrt{2}^{2n}}\langle \mathbf{x}^{\sigma} \mathbf{x}^{\tau} | H^{\otimes 2n} e^{-i \sum_{ij} J_{ij}\sigma^{(i)}_z \tau^{(j)}_z t/n} \sum_{\w^\sigma \in \{0,1\}^{n}} \ket{\w^\sigma} \otimes \sum_{\w^\tau \in \{0,1\}^{n}} (-1)^{ \wt(\w^{\tau}) } \ket{\w^\tau} \\
    &= \frac{1}{\sqrt{2}^{2n}} \sum_{\substack{\w^\sigma \in \{0,1\}^{n}\\ \w^\tau \in \{0,1\}^{n}}} (-1)^{ \wt(\w^{\tau}) }e^{-i \sum_{ij} J_{ij}/n (1 - 2 w^{\sigma}_i)(1-2 w^{\tau}_j) t}  \langle \mathbf{x}^{\sigma} \mathbf{x}^{\tau} | H^{\otimes 2n}  \ket{\w^\sigma \w^\tau} \numberthis .
\end{align*}

From
\begin{align*}
    \langle \mathbf{x}^{\sigma} \mathbf{x}^{\tau} | H^{\otimes 2n}  \ket{\w^\sigma \w^\tau} &= \frac{1}{\sqrt{2}^{2n}} \sum_{\y^{\sigma} \in \{0,1\}^n} \sum_{\y^{\tau} \in \{0,1\}^n} (-1)^{\x^{\sigma} \cdot \y^\sigma + \x^{\tau} \cdot \y^\tau }\delta_{\y^\sigma,\w^\sigma} \delta_{\y^\tau,\w^\tau} \\
    &= \frac{1}{\sqrt{2}^{2n}} (-1)^{\w^\sigma \cdot \x^\sigma + \w^\tau \cdot \x^\tau}, \numberthis
\end{align*}
we obtain
\begin{align*}
    f(\mathbf{x};J;t) &= \frac{1}{2^{2n}} \sum_{\substack{\w^\sigma \in \{0,1\}^{n}\\ \w^\tau \in \{0,1\}^{n}}}  (-1)^{\wt(\w^{\tau})  + \w^\sigma \cdot \x^\sigma + \w^\tau \cdot \x^\tau} e^{-i \sum_{ij} J_{ij}/n (1 - 2 y^{\sigma}_i)(1-2 y^{\tau}_j) t}. \numberthis
\end{align*}

Then, the output probability is given by
\begin{align*}
    p(\mathbf{x};J;t) &= f(\mathbf{x};J;t) f(\mathbf{x};J;t)^* \\
    &= \frac{1}{2^{4n}} \sum_{\substack{\w^\sigma \in \{0,1\}^{n}\\ \w^\tau \in \{0,1\}^{n}}} \sum_{\substack{\z^\sigma \in \{0,1\}^{n}\\ \z^\tau \in \{0,1\}^{n}}} (-1)^{\wt(\w^{\tau})  + \w^\sigma \cdot \x^\sigma + \w^\tau \cdot \x^\tau + \wt(\z^{\tau})  + \z^\sigma \cdot \x^\sigma + \z^\tau \cdot \x^\tau} \\
    &\qquad \qquad \times e^{-i \sum_{ij} J_{ij}/n [ (1 - 2 w^{\sigma}_i)(1-2 w^{\tau}_j) - (1 - 2 z^{\sigma}_i)(1-2 z^{\tau}_j)] t}. \numberthis
\end{align*}

We now assume that $\{J_{ij}\}$ is rationally independent, i.e., 
\begin{align}
    \sum_{ij} J_{ij} a_{ij} = 0 \text{ for } a_{ij} \in \mathbb{Z}, \text{ iff } a_{ij}=0 \text{ for all } i,j.
\end{align}
This condition is satisfied for a randomly chosen $J_{ij}$.

Then, it follows that
\begin{align}
    \lim_{T \rightarrow \infty} \frac{1}{T} \int_{0}^T dt e^{-i \sum_{ij} J_{ij}/n[ (1 - 2 w^{\sigma}_i)(1-2 w^{\tau}_j) - (1 - 2 z^{\sigma}_i)(1-2 z^{\tau}_j)] t} = \begin{cases}
        1 \text{ if } \w=\z \text{ or } \w = \overline{\z} \\
        0 \text{ otherwise}
    \end{cases},
\end{align}
where $\overline{\z} = \{1-z_i\}$ is a vector after taking binary complement for each component of $\z$. 

We now take the time average of $p(\mathbf{x};J;t)$, which yields
\begin{align*}
    \overline{p}(\mathbf{x};J) &= \lim_{T \rightarrow \infty} \frac{1}{T} \int_{0}^T dt  p(\mathbf{x};J;t) \\
    &= \frac{1}{2^{4n}} \Biggl[ \sum_{\substack{\w^\sigma \in \{0,1\}^{n} \\ \w^\tau \in \{0,1\}^{n}}} (-1)^{2\wt(\w^{\tau})  + 2\w^\sigma \cdot \x^\sigma + 2\w^\tau \cdot \x^\tau} + \sum_{\substack{\w^\sigma \in \{0,1\}^{n} \\ \w^\tau \in \{0,1\}^{n}}} (-1)^{\wt(\w^{\tau})  + \w^\sigma \cdot \x^\sigma + \w^\tau \cdot \x^\tau + \wt(\overline{\w}^{\tau})  + \overline{\w}^\sigma \cdot \x^\sigma + \overline{\w}^\tau \cdot \x^\tau} \Biggr] \\
    &= \frac{1}{2^{2n}} \Biggl[ 1 + (-1)^{n + \wt(\x^\sigma) + \wt(\x^\tau)} \Biggr]  = \frac{1}{2^{2n}} \Biggl[ 1 + (-1)^{n + \wt(\x)} \Biggr], \numberthis
\end{align*}
where we have used $\wt(\w^\tau) + \wt(\overline{\w}^\tau) = |\w^{\tau}| = n$ (and the same for $\w^{\sigma}$) and $\w^\tau \cdot \x^\tau + \overline{\w}^\tau \cdot \x^\tau = \mathbf{1}^\tau \cdot \x^\tau = \wt(\x^\tau)$ (and the same for $\x^\sigma$).

The final expression implies that $p(\mathbf{x};J;t)$ converges to $2 \times 2^{-2n}$ or $0$ as $t$ goes to infinity, depending on the Hamming weight of $\x$.
In addition, the converged values do not depend on $J_{ij}$.
Given that $\x \in X$ has the Hamming weight of $n$, we know that $p(\mathbf{x};J;t)$ converges to $2 \times 2^{-2n}$, which is consistent with the numerical observation from the previous subsection.

We finally note that this convergence behavior also supports anticoncentration as it implies that the probability distribution, $p(\mathbf{x};J;t)$, after an equilibration time, $t_{\rm eq}$, only fluctuates near $2 \times 2^{-2n}$ regardless of $J$ and $\mathbf{x}$.

\section{Estimating the number of required gates for the XY model} \label{app:estimating_number_of_gates}
In the main text, we presented the number of gates required to simulate the XY model.
For $n=100$ and $t_0 = 5 \log(n) \approx 23.03$, we estimated that the number of required gates is between around $10^{8}$ to $10^9$ depending on the target Trotter error.
This appendix provides detailed steps on how these estimations are obtained.

Our main tool is the following inequality:
\begin{align*}
    &\Biggl\Vert \exp \biggl[ -i t_0\sum_{ij} \frac{J_{ij}}{2n} \Bigl( \sigma^{(i)}_x \tau^{(j)}_x + \sigma^{(i)}_y \tau^{(j)}_y \Bigr) \biggr] - \mathcal{T}_2^M \biggr] \Biggr\Vert \leq \mathcal{O} \left[ M \Bigl( \sum_{ij} \frac{|J_{ij}|}{n} \frac{t_0}{M} \Bigr)^3 \right],
\end{align*}
where $\mathcal{T}_2$ is the expression obtained by applying the second-order Suzuki--Trotter formula for each Trotter step (for a timestep of $t_0/M$).
Thus, $\mathcal{T}_2$ can be implemented using a quantum circuit with $2n^2$ gates.

When $n$ is sufficiently large, we can change $\sum_{ij} |J_{ij}|$ to its average value given by $n^2 \sqrt{2/\pi}$. The total Trotter error is then approximated by $\epsilon_t \approx P \frac{n^3 t_0^3 }{M^2}$ where $P$ is a pre-factor.
The prefactor $P$ is estimated to $\approx 2.97 \times 10^{-4}$ by computing the distance (measured w.r.t. the operator norm) between the unitary operator, $e^{-iHt}$, and the Trotterized expression, $\mathcal{T}_2^M$, for $n=5$.

Using this value of $P$, results for $n=100$ are extrapolated.
For the target Trotter error $\epsilon_t$, the number of gates is estimated as
\begin{align}
    \#(\text{gates}) = 2n^2M = 2n^2\sqrt{\frac{P n^3 t_0^3}{\epsilon_t}}.
\end{align}
By entering $P\approx 2.97 \times 10^{-4}$, $n=10^2$, and $t_0\approx 23.03$, we obtained the number of gates $1.2 \times 10^{8}$, $3.8 \times 10^{8}$, and $1.2 \times 10^{9}$ for $\epsilon_t = 10^{-1}$, $10^{-2}$, and $10^{-3}$, respectively.

\end{document}